\documentclass{IEEEtaes}

\usepackage[colorlinks,urlcolor=blue,linkcolor=blue,citecolor=blue]{hyperref}

\usepackage{color,array,amsthm}
\usepackage{cite}
\usepackage{amsmath,amssymb,amsfonts}
\usepackage{algorithmic, algorithm}
\usepackage{graphicx,subfigure}
\usepackage{textcomp}
\usepackage{graphicx}
\usepackage{amsmath}
\usepackage{xcolor}
\usepackage{bm}
\usepackage{diagbox}
\usepackage{multirow}

\newtheorem{lemma}{Lemma}

\usepackage{color,soul}

\newcommand{\eref}[1]{(\ref{#1})}
\newcommand{\sref}[1]{Section~\ref{#1}}

\newcommand{\fref}[1]{Figure~\ref{#1}}

\newcommand{\cref}[1]{Constraint~\ref{#1}}

\newcommand{\lref}[1]{Lemma~\ref{#1}}

\newcommand{\tref}[1]{Table~\ref{#1}}

\newcommand{\algref}[1]{Algorithm~\ref{#1}}

\begin{document}

\title{Instantaneous GNSS Ambiguity Resolution and Attitude Determination via Riemannian Manifold Optimization}

\author{Xing Liu}

\author{Tarig Ballal}
\member{Member, IEEE}

\author{Mohanad Ahmed}

\author{Tareq Y. Al-Naffouri}
\member{Senior Member, IEEE}
\affil{King Abdullah University of Science and Technology (KAUST), Thuwal, Saudi Arabia}


\authoraddress{The authors are with the Division of Computer, Electrical and Mathematical Sciences, and Engineering, King Abdullah University of Science and Technology, Thuwal 23955-6900, Saudi Arabia (e-mail: xing.liu@kaust.edu.sa; tarig.ahmed@kaust.edu.sa; mohanad.ahmed@kaust.edu.sa; tareq.alnaffouri@kaust.edu.sa).}


\maketitle

\begin{abstract}
We present an ambiguity resolution method for Global Navigation Satellite System (GNSS)-based attitude determination. A GNSS attitude model with nonlinear constraints is used to rigorously incorporate a priori information. Given the characteristics of the employed nonlinear constraints, we formulate GNSS attitude determination as an optimization problem on a manifold. Then, Riemannian manifold optimization algorithms are utilized to aid ambiguity resolution based on a proposed decomposition of the objective function. The application of manifold geometry enables high-quality float solutions that are critical to reinforcing search-based integer ambiguity resolution in terms of efficiency, availability, and reliability. The proposed approach is characterized by a low computational complexity and a high probability of resolving the ambiguities correctly. The performance of the proposed ambiguity resolution method is tested through a series of simulations and real experiments. Comparisons with the principal benchmarks indicate the superiority of the proposed method as reflected by the high ambiguity resolution success rates.
\end{abstract}


\

\

\

\

\

\

\

\section{INTRODUCTION}
Global navigation satellite system (GNSS)-based attitude determination has recently received much attention as it plays a critical role in various navigation, guidance, and control applications \cite{9292084,9050874,ardalan2015iterative, 5975219, 8869594,hofmann2007gnss}. Nowadays, GNSS attitude determination has been ubiquitously explored in numerous land, airborne, and maritime scenarios \cite{LiJul2004, 4408595,9744503,6129643,ChiOct2014, 6200891}. The goal of GNSS attitude determination is to estimate a vehicle's orientation with respect to a reference coordinate \cite{hofmann2012global,hofmann2007gnss}, utilizing multiple GNSS antennas/receivers rigidly mounted on the body frame. Compared with the other sensors applied in attitude determination, such as gyroscopes or star-trackers, a GNSS receiver offers several advantages, including the driftless characteristic, low power consumption, and minor maintenance requirements \cite{madsen2004robust, crassidis1997new}. 

A GNSS receiver can generate pseudo-range and carrier-phase observations based on the signals transmitted from the navigation satellites. The measurement accuracy of the pseudo-range observations is two orders of magnitude lower than that of carrier-phase observations \cite{giorgi2013low}. It is essential to employ the precise GNSS data, the carrier phase, to perform a high-accuracy attitude estimate. The main challenge to fully utilize carrier-phase observations is to successfully resolve the unknown integer parts (number of whole cycles), a process usually referred to as integer ambiguity resolution. The integer ambiguities can be resolved by employing the change of the receiver-satellite geometry provoked by the vehicle's movement in the techniques known as motion-based methods \cite{crassidis1999global, ChuMay2001, wang2010motion}. However, this class of methods demands multi-epoch GNSS data such that they can not meet the requirement of real-time applications. More recently, instantaneous ambiguity resolution approaches have been developed to solve the problem by searching in the float (attitude/position) or integer domains \cite{teunissen1993least,chang2005mlambda,LiJul2004,PurApr2010, 9110131,8955972,liu2018gnss, 9650516,liu2021attitude}, requiring only a single-epoch GNSS observation.

The least-squares ambiguity decorrelation adjustment (LAMBDA) method \cite{teunissen1993least,teunissen1994new} has become the standard approach for GNSS ambiguity resolution with unconstrained or linearly constrained models \cite{teunissen2012affine, giorgi2013low, 6491499}. It can efficiently recover the ambiguities by searching for the integers within an ellipsoidal region. For unconstrained or linearly constrained models, LAMBDA solutions are optimal in terms of the ambiguity resolution success rate, i.e., it offers the highest probability of resolving the ambiguities correctly \cite{teunissen1999optimality}. However, this technique fails to take advantage of a priori knowledge of antenna geometry, which introduces nonlinear (orthonormality) constraints to the optimization problem.

Constrained LAMBDA (C-LAMBDA) \cite{teunissen2010integer} and multivariate constrained LAMBDA (MC-LAMBDA) methods \cite{giorgi2012instantaneous} were proposed to leverage the priori knowledge about the antenna relative positions for single- and multi-baseline attitude models, respectively. The literature shows that these two techniques can significantly improve the success-rate performance compared to the LAMBDA method \cite{teunissen2010integer, 6491499}. Resolving the unknown carrier-phase ambiguities correctly allows for accurate attitude determination, whereas the orthonormality constraints facilitate high ambiguity resolution success rates and reliable attitude estimation \cite{teunissen2012affine}. Incorporating the nonlinear constraints into the optimization problem results in a more complex search space that is not ellipsoidal anymore, increasing the complexity of the search process significantly \cite{giorgi2010carrier, 6491499}. The performance of the C-LAMBDA or MC-LAMBDA approaches relies upon the quality of the (initial) float solutions, that is, the least-squares estimations with the integer constraint of ambiguities ignored. Under satellite-deprived environments such as urban canyons, searching the integers requires a huge computational load due to the poor float solutions initially obtained through unconstrained least-squares (LS) estimation \cite{liu2021constrained}.

Based on a specific relaxation of the nonlinear constraints pertaining to the antenna-array geometry, the affine-constrained attitude model (AC-AM) has been proposed as a compromise between the unconstrained LAMBDA and MC-LAMBDA methods \cite{teunissen2012affine, giorgi2013low}. The AC-AM relaxes the orthonormality constraints to linear ones through an affine transformation. Therefore, if the number of baselines is greater than the dimension of the range of the baseline matrix, the AC-AM provides a better float solution compared to the unconstrained model such that it outperforms the standard LAMBDA in the terms of the ambiguity resolution success rate\cite{teunissen2012affine}. On the other hand, the AC-AM is less robust than the orthonormality-constrained attitude model (OC-AM) used in the MC-LAMBDA method; thus, high computational efficiency but lower success rates followed. Given a large number of baselines, this method can generally provide high success rate, allowing it to be a fast alternative to the MC-LAMBDA method \cite{giorgi2013low}. However, with a limited number of baselines, there is still much room for success rate improvement.

In this article, we propose an ambiguity resolution method for GNSS attitude determination via Riemannian manifold optimization. As the solution set of the attitude matrix is a manifold, we formulate GNSS attitude determination as an optimization over Riemannian manifolds. An earlier study of manifold geometries has demonstrated the ability to develop efficient Riemannian algorithms with the capacity of offering accurate attitude estimations using the given carrier-phase ambiguities \cite{Ahmed2020RieOpt}. 
We demonstrate that Riemannian manifold optimization can also be utilized in the ambiguity resolution process to significantly improve the float solution. In \sref{sec: IntSol}, we will demonstrate how imposing the geometric constraints can improve not only the rotation matrix estimation, but also the float estimates of the ambiguities. This performance gain is attributed to the improved robust model adopted to carry out the initial float estimation operation, which assists the integer ambiguity search operation that follows.

To refine the search efficiency, a decomposition of the objective function is proposed that also leverages the improved float solution, which shrinks the search space and accelerates the ambiguity resolution process. At the expense of the computation burden of one manifold optimization, the loose form of the proposed approach can beat the AC-AM in terms of success rate, and its tight-form can outperform the MC-LAMBDA method in terms of computational efficiency. To the best of the author's knowledge, this is the first work to apply manifold optimization to assist the integer ambiguity resolution process in GNSS attitude determination.

The structure of this paper is as follows. \sref{sec: model} describes the various GNSS attitude models and the related work in the literature. \sref{sec:ove} briefly introduces the necessary ingredients for optimization algorithms on Riemannian manifolds. \sref{sec:proposed} presents the proposed attitude determination algorithm. \sref{sec: IntSol} addresses the formulation of the intermediate float solution used to improve ambiguity resolution, and presents in detail the proposed method to obtain the float solution using Riemannian manifold optimization. \sref{sec:ambres} describes the proposed ambiguity resolution method aided by the float solution, including a proposed decomposition and search strategy. Finally, \sref{sec: SimExp} presents numerical results, comparing the performance of the proposed approach and a number of benchmark methods using different evaluation criteria.

\section{Related Work}
\label{sec: model}
\subsection{The Unconstrained Attitude Determination}
Consider a vehicle equipped with $\mathcal{A}+1$ GNSS antennas that track $\mathcal{S}+1$ satellites simultaneously. The double-difference (DD) operation is applied to cancel out the common-mode errors, such as the receiver and satellite clock biases, hardware offsets, and atmospheric delays. In a single-frequency and single-epoch case, the GNSS attitude observation equations are given by
\begin{equation}
\begin{aligned}
  {\bm{\Psi}} &= {\mathbf{H}}{\mathbf{X}} + \mathbf{W}\mathbf{N} + \bm{\Pi},\\
  {\mathbf{P}} &= {\mathbf{H}}\mathbf{X} + \bm{\Xi},
  \end{aligned}
  \label{eq:model1}
\end{equation}
where ${\bm{\Psi}} \in {\mathbb{R}^{\mathcal{S} \times \mathcal{A}}}$ and ${\mathbf{P}} \in {\mathbb{R}^{\mathcal{S} \times \mathcal{A}}}$ denote the DD carrier-phase and pseudo-range matrices, respectively, ${\mathbf{H}} \in {\mathbb{R}^{\mathcal{S} \times 3}}$ is composed of the satellite line-of-sight vectors, ${\mathbf{W}} \in {\mathbb{R}^{\mathcal{S} \times \mathcal{S}}}$ contains the wavelength of the GNSS frequency, the columns of ${\mathbf{X}} \in {\mathbb{R}^{3 \times \mathcal{A}}}$ represent the unknown baseline coordinates in the chosen reference frame, $\mathbf{N} \in {\mathbb{Z}^{\mathcal{S} \times \mathcal{A}}}$ denotes the unknown DD integer ambiguities, and $\bm{\Pi} \in {\mathbb{R}^{\mathcal{S} \times \mathcal{A}}}$ and $\bm{\Xi} \in {\mathbb{R}^{\mathcal{S} \times \mathcal{A}}}$ are unmodelled errors and noise. 

The model \eref{eq:model1} is called the unconstrained attitude model (UC-AM) since it disregards the prior knowledge of the baseline matrix. Note that ``unconstrained'' is only in terms of the antenna-array geometry; the integer constraints are given full consideration. To estimate the unknown baseline coordinates and integer ambiguities, the problem can be formulated as
\begin{equation}
\mathop {\min }\limits_{{\mathbf{X}} \in {\mathbb{R}^{3 \times \mathcal{A}}}, {\mathbf{N}} \in {\mathbb{Z}^{\mathcal{S} \times\mathcal{A}}}}\left\| {\operatorname{vec}\left( {{\mathbf{Y}} - {\mathbf{A}}{\mathbf{X}} - {\mathbf{BN}}} \right)} \right\|_{\mathbf{Q}_{\mathbf{Y\!Y}}^{-1}}^2,
\label{eq:unam}
\end{equation}
where
\[{\mathbf{Y}} \triangleq \left[ {\begin{array}{*{20}{c}}
  {\bm{\Psi}} \\ 
  {\bf{P}}
\end{array}} \right],
{\mathbf{A}} \triangleq \left[ {\begin{array}{*{20}{c}}
  {\mathbf{H}} \\ 
  {\mathbf{H}} 
\end{array}} \right],
{\mathbf{B}} \triangleq \left[ {\begin{array}{*{20}{c}}
  \mathbf{W} \\ 
  \mathbf{O}
\end{array}} \right].\]
The matrices ${\mathbf{A}}$ and ${\mathbf{B}}$ link the GNSS measurements and the unknown parameters. The notation $\operatorname{vec}(\cdot)$ represents the vectorization operation, $\left\| (\cdot) \right\|_{{\mathbf{Q}}_{\mathbf{Y\!Y}}^{-1}}^2 = (\cdot)^{\text{T}}{\mathbf{Q}}_{\mathbf{Y\!Y}}^{-1}(\cdot)$, and ${\mathbf{Q}_{\mathbf{Y\!Y}}}$ is the covariance matrix of $\operatorname{vec}({\mathbf{Y}})$, that is
\begin{equation*}
\begin{aligned}
{\mathbf{Q}_{\mathbf{Y\!Y}}} \!=\! \mathbb{E}\!\!\left[ \operatorname{vec}({\mathbf{Y}} \!-\! \mathbb{E}({\mathbf{Y}}) )\left[\operatorname{vec}({\mathbf{Y}} \!-\! \mathbb{E}({\mathbf{Y}}) )\right]^{\text{T}}\right],
\end{aligned}
\end{equation*}
where $\mathbb{E}( \cdot )$ is the expectation operator.

To determine the solution of \eref{eq:unam}, the following orthogonal decomposition can be applied to simplify the problem \cite{giorgi2013low}:
\begin{equation}
\resizebox{.9\hsize}{!}{
$
\begin{aligned}
&\mathop {\min }\limits_{{\mathbf{X}} \in {\mathbb{R}^{3 \times \mathcal{A}}}, {\mathbf{N}} \in {\mathbb{Z}^{\mathcal{S} \times\mathcal{A}}}} \!\left\| {\operatorname{vec} \! \left( {{\mathbf{Y}} \!-\! {\mathbf{A}}{\mathbf{X}} \!-\! {\mathbf{BN}}} \right)}\! \right\|_{\mathbf{Q}_{\mathbf{Y\!Y}}^{-1}}^2
= \left\| {\operatorname{vec} \! \left( {\mathbf{E}}_{_\text{UC}} \right)}\! \right\|_{\mathbf{Q}_{\mathbf{Y\!Y}}^{-1}}^2 \\ &+ \!\mathop {\min }\limits_{{\mathbf{N}} \in {\mathbb{Z}^{\mathcal{S} \times\mathcal{A}}}} \!\! \left(\!\! \left\| {\operatorname{vec} \! \left( \! {\mathbf{N}} \! - \! {\mathbf{\hat N}_{_\text{UC}}} \!\right)} \!\right\|_{\mathbf{Q}_{{\mathbf{\hat N}_{_\text{UC}}\!\mathbf{\hat N}_{_\text{UC}}}}^{-1}}^2\! \!+\!
\mathop {\min }\limits_{{\mathbf{X}} \in {\mathbb{R}^{3 \times \mathcal{A}}}}\!\left\| {\operatorname{vec} \!\! \left(\!{\mathbf{X}} \!-\! {\mathbf{\hat X}_{_\text{UC}}}\!\!\left(\!{\mathbf{N}}\!\right) \!\!\right)} \!\right\|_{\mathbf{Q}_{{\mathbf{\hat X}_{_\text{UC}}}\!\left(\!{\mathbf{N}}\!\right){\mathbf{\hat X}_{_\text{UC}}}\!\left(\!{\mathbf{N}}\!\right)}^{-1}}^2 \!\right).
\label{eq:decom1}
\end{aligned}
$}
\end{equation}
In \eref{eq:decom1}, ${\mathbf{E}}_{_\text{UC}} = {{\mathbf{Y}} \!-\! {\mathbf{A}}{\mathbf{\hat X}}_{_\text{UC}}  \!-\! {\mathbf{B \hat N}}_{_\text{UC}} }$ is the residual matrix of the least-squares solution, i.e., the float solution without the integer constraints given by
\begin{equation}
\resizebox{.9\hsize}{!}{
$
\begin{aligned}
{{{\mathbf{\hat X}}}_{_\text{UC}} }, {{{\mathbf{\hat N}}}_{_\text{UC}} }\! \! =\! \!\mathop {\arg\min }\limits_{{\mathbf{X}} \in {\mathbb{R}^{3 \times \mathcal{A}}}, {\mathbf{N}} \in {\mathbb{R}^{\mathcal{S} \times\mathcal{A}}}} \!\left\| {\operatorname{vec} \! \left( {{\mathbf{Y}} \!-\! {\mathbf{A}}{\mathbf{X}} \!-\! {\mathbf{BN}}} \right)}\! \right\|_{\mathbf{Q}_{\mathbf{Y\!Y}}^{-1}}^2,
\label{eq:float1}
\end{aligned}
$}
\end{equation}
whose covariance matrixes are $\mathbf{Q}_{_{{\mathbf{\hat N}_{_{\text{UC}}}}\!{\mathbf{\hat N}_{_\text{UC}}}}}$, $\mathbf{Q}_{_{{\mathbf{\hat X}_{_\text{UC}}}\!{\mathbf{\hat X}_{_\text{UC}}}}}$, $\mathbf{Q}_{_{{\mathbf{\hat N}_{_\text{UC}}}\!{\mathbf{\hat X}_{_\text{UC}}}}}$, and $\mathbf{Q}_{_{{\mathbf{\hat X}_{_\text{UC}}}\!{\mathbf{\hat N}_{_\text{UC}}}}}$. The other variables in \eref{eq:decom1} are expressed as
\begin{equation}
\resizebox{.9\hsize}{!}{
$
\begin{aligned}
&\operatorname{vec} \!\!  \left(\! {\mathbf{\hat X}_{_\text{UC}}}\!\!\left(\!{\mathbf{N}}\!\right)\!\!  \right) \!= \!\operatorname{vec} \!\!  \left(\!{\mathbf{\hat X}_{_\text{UC}}}\! \right) \! -\!  \mathbf{Q}_{_{{\mathbf{\hat X}_{_\text{UC}}}\!{\mathbf{\hat N}_{_\text{UC}}}}}\! \!  \mathbf{Q}_{_{{\mathbf{\hat N}_{_{\text{UC}}}}\!{\mathbf{\hat N}_{_\text{UC}}}}}^{-1}\!\! \operatorname{vec} \! \!\left( \! \mathbf{\hat N}_{_{\text{UC}}} \!\! -\!  \mathbf{N} \!\!  \right)\!,\\
&\mathbf{Q}_{_{{\mathbf{\hat X}_{_\text{UC}}}\!\left({\mathbf{N}}\right){\mathbf{\hat X}_{_\text{UC}}}\!\left({\mathbf{N}}\right)}} \!\! \!=\! \mathbf{Q}_{_{{\mathbf{\hat X}_{_\text{UC}}}\!{\mathbf{\hat X}_{_\text{UC}}}}} \!\!-\! \mathbf{Q}_{_{{\mathbf{\hat X}_{_\text{UC}}}\!{\mathbf{\hat N}_{_\text{UC}}}}}\!\! \mathbf{Q}_{_{{\mathbf{\hat N}_{_{\text{UC}}}}\!{\mathbf{\hat N}_{_\text{UC}}}}}^{-1} \!\! \mathbf{Q}_{_{{\mathbf{\hat N}_{_\text{UC}}}\!{\mathbf{\hat X}_{_\text{UC}}}}}.
\end{aligned}
$}
\end{equation}
For further details, refer to \cite{giorgi2013low}.

The residual term is independent of the parameters. Since the UC-AM disregards the constraints on $\mathbf{X}$, we can let ${\mathbf{X}} = {\mathbf{\hat X}_{_\text{UC}}}\!\!\left(\!{\mathbf{N}}\!\right)$ so that the third term on the right-hand side of \eref{eq:decom1} is equal to zero, i.e., minimizing \eref{eq:decom1}. Hence, the original optimization in \eref{eq:unam} is equivalent to the following integer least-squares (ILS) problem
\begin{equation}
    \mathop {\min }\limits_{{\mathbf{N}} \in {\mathbb{Z}^{\mathcal{S} \times\mathcal{A}}}} \! \left\| {\operatorname{vec} \! \left( \! {\mathbf{N}} \! - \! {\mathbf{\hat N}_{_\text{UC}}} \!\right)} \!\right\|_{\mathbf{Q}_{{\mathbf{\hat N}_{_\text{UC}}\!\mathbf{\hat N}_{_\text{UC}}}}^{-1}}^2.
    \label{eq:lambda1}
\end{equation}
The LAMBDA method can solve \eref{eq:lambda1} efficiently, with ${\mathbf{\hat N}_{_\text{UC}}}$ being the center of the search space (an ellipsoidal set). Once the integer ambiguities are resolved, the recovered DD carrier phase can be utilized to estimate the baseline coordinates.

\subsection{The Affine-constrained Attitude Determination}
One can accurately measure the antenna-array geometry because the GNSS antennas are firmly mounted on the vehicle. Therefore, the antenna-array coordinates in the body frame can be considered as known parameters that can be incorporated in GNSS attitude determination. In \cite{teunissen2012affine, giorgi2013low}, to leverage the priori knowledge, the authors take advantage of an affine transformation, which converts the baseline coordinates from the body frame to the reference frame through
\begin{equation}
{{\mathbf{X}}} = {{\mathbf{R}}}{{\mathbf{X}_b}}, \quad {\mathbf{R}} \in {\mathbb{R}^{3 \times q}},
\label{eq:link}
\end{equation}
where ${\mathbf{X}_b} \in {\mathbb{R}^{q \times \mathcal{A}}}$, and $q = \operatorname{min}(3,\mathcal{A})$. Then, \eref{eq:model1} and \eref{eq:link} establish the affine-constrained attitude model (AC-AM).

Based on the affine transformation, the attitude determination problem can be expressed as \begin{equation}
\mathop {\min }\limits_{{\mathbf{R}} \in {\mathbb{R}^{3 \times \mathcal{A}}}, {\mathbf{N}} \in {\mathbb{Z}^{\mathcal{S} \times\mathcal{A}}}}\left\| {\operatorname{vec}\left( {{\mathbf{Y}} - {\mathbf{AR}}{\mathbf{X}_b} - {\mathbf{BN}}} \right)} \right\|_{\mathbf{Q}_{\mathbf{Y\!Y}}^{-1}}^2.
\label{eq:acam}
\end{equation}
Similar to \eref{eq:decom1}, the following orthogonal decomposition is used 
\begin{equation}
\resizebox{.90\hsize}{!}{
$
\begin{aligned}
&\mathop {\min }\limits_{{\mathbf{R}} \in {\mathbb{R}^{3 \times q}}, {\mathbf{N}} \in {\mathbb{Z}^{\mathcal{S} \times\mathcal{A}}}} \!\left\| {\operatorname{vec} \! \left( {{\mathbf{Y}} \!-\! {\mathbf{AR}}{\mathbf{X}_b} \!-\! {\mathbf{BN}}} \right)}\! \right\|_{\mathbf{Q}_{\mathbf{Y\!Y}}^{-1}}^2
= \left\| {\operatorname{vec} \! \left( {\mathbf{E}}_{_\text{AC}} \right)}\! \right\|_{\mathbf{Q}_{\mathbf{Y\!Y}}^{-1}}^2 \\ &+ \!\mathop {\min }\limits_{{\mathbf{N}} \in {\mathbb{Z}^{\mathcal{S} \times\mathcal{A}}}} \!\! \left(\!\! \left\| {\operatorname{vec} \! \left( \! {\mathbf{N}} \! - \! {\mathbf{\hat N}_{_\text{AC}}} \!\right)} \!\right\|_{\mathbf{Q}_{{\mathbf{\hat N}_{_\text{AC}}\!\mathbf{\hat N}_{_\text{AC}}}}^{-1}}^2\!\!\! +\!
\mathop {\min }\limits_{{\mathbf{R}} \in {\mathbb{R}^{3 \times q}}}\!\left\| {\operatorname{vec} \!\! \left(\!{\mathbf{R}} \!-\! {\mathbf{\hat R}_{_\text{AC}}}\!\!\left(\!{\mathbf{N}}\!\right)\! \!\right)} \!\right\|_{\mathbf{Q}_{{\mathbf{\hat R}_{_\text{AC}}}\!\left({\mathbf{N}}\right){\mathbf{\hat R}_{_\text{AC}}}\!\left({\mathbf{N}}\right)}^{-1}}^2 \!\right),
\label{eq:decom2}
\end{aligned}
$}
\end{equation}
where ${\mathbf{E}}_{_\text{AC}} = {{\mathbf{Y}} \!-\! {\mathbf{A}}{\mathbf{\hat R}}_{_\text{AC}}{\mathbf{X}_b}  \!-\! {\mathbf{B \hat N}}_{_\text{AC}} }$ is the least-squares residual matrix of the float solution given by
\begin{equation}
\resizebox{.86\hsize}{!}{
$
\begin{aligned}
{{{\mathbf{\hat R}}}_{_\text{AC}} }, {{{\mathbf{\hat N}}}_{_\text{AC}} }\! \! =\! \!\mathop {\arg\min }\limits_{{\mathbf{R}} \in {\mathbb{R}^{3 \times \mathcal{A}}}, {\mathbf{N}} \in {\mathbb{R}^{\mathcal{S} \times\mathcal{A}}}} \!\left\| {\operatorname{vec} \! \left( {{\mathbf{Y}} \!-\! {\mathbf{AR}}{\mathbf{X}_b} \!-\! {\mathbf{BN}}} \right)}\! \right\|_{\mathbf{Q}_{\mathbf{Y\!Y}}^{-1}}^2,
\label{eq:float2}
\end{aligned}
$}
\end{equation}
with the covariance matrix $\mathbf{Q}_{_{{\mathbf{\hat N}_{_{\text{AC}}}}\!{\mathbf{\hat N}_{_\text{AC}}}}}$, $\mathbf{Q}_{_{{\mathbf{\hat R}_{_\text{AC}}}\!{\mathbf{\hat R}_{_\text{AC}}}}}$, $\mathbf{Q}_{_{{\mathbf{\hat N}_{_\text{AC}}}\!{\mathbf{\hat R}_{_\text{AC}}}}}$, and $\mathbf{Q}_{_{{\mathbf{\hat R}_{_\text{AC}}}\!{\mathbf{\hat N}_{_\text{AC}}}}}$. The other variables in \eref{eq:decom2} are given by
\begin{equation}
\resizebox{.86\hsize}{!}{
$
\begin{aligned}
&\operatorname{vec} \!\!  \left(\!\!{\mathbf{\hat R}_{_\text{AC}}}\!\!\left(\!{\mathbf{N}}\!\right)\!\! \right) \!= \!\operatorname{vec} \!\!  \left(\!{\mathbf{\hat R}_{_\text{AC}}}\! \right) \!-\! \mathbf{Q}_{_{{\mathbf{\hat R}_{_\text{AC}}}\!{\mathbf{\hat N}_{_\text{AC}}}}}\!\! \mathbf{Q}_{_{{\mathbf{\hat N}_{_{\text{AC}}}}\!{\mathbf{\hat N}_{_\text{AC}}}}}^{-1}\!\!\operatorname{vec} \!\!  \left(\!\!\mathbf{\hat N}_{_{\text{AC}}}\! \!- \!\mathbf{N}\!\!\right)\!,\\
&\mathbf{Q}_{_{{\mathbf{\hat R}_{_\text{AC}}}\!\left(\!{\mathbf{N}}\!\right){\mathbf{\hat R}_{_\text{AC}}}\!\left(\!{\mathbf{N}}\!\right)}} \!\! \!=\! \mathbf{Q}_{_{{\mathbf{\hat R}_{_\text{AC}}}\!{\mathbf{\hat R}_{_\text{AC}}}}} \!\!-\! \mathbf{Q}_{_{{\mathbf{\hat R}_{_\text{AC}}}\!{\mathbf{\hat N}_{_\text{AC}}}}}\!\! \mathbf{Q}_{_{{\mathbf{\hat N}_{_{\text{AC}}}}\!{\mathbf{\hat N}_{_\text{AC}}}}}^{-1} \!\! \mathbf{Q}_{_{{\mathbf{\hat N}_{_\text{AC}}}\!{\mathbf{\hat R}_{_\text{AC}}}}}.
\end{aligned}
$}
\end{equation}
Finally, the optimization problem in \eref{eq:acam} is identical to the following minimization \cite{teunissen2012affine}
\begin{equation}
    \mathop {\min }\limits_{{\mathbf{N}} \in {\mathbb{Z}^{\mathcal{S} \times\mathcal{A}}}} \! \left\| {\operatorname{vec} \! \left( \! {\mathbf{N}} \! - \! {\mathbf{\hat N}_{_\text{AC}}} \!\right)} \!\right\|_{\mathbf{Q}_{{\mathbf{\hat N}_{_\text{AC}}\!\mathbf{\hat N}_{_\text{AC}}}}^{-1}}^2.
    \label{eq:lambda2}
\end{equation}
Since $\mathbf{Q}_{_{{\mathbf{\hat N}_{_{\text{AC}}}}\!{\mathbf{\hat N}_{_\text{AC}}}}} \!\!< \!\mathbf{Q}_{_{{\mathbf{\hat N}_{_{\text{UC}}}}\!{\mathbf{\hat N}_{_\text{UC}}}}}$, the AC-AM can improve the float solution such that it can provide better ambiguity estimations compared with the UC-AM \cite{giorgi2013low}.

\subsection{The Orthonormality-Constrained Attitude Determination}
\label{sec: ocmodel}
The AC-AM does not rigorously integrate the known geometry of the GNSS antenna configuration into the objective function. To fully leverage the nonlinear constraints of antenna-array geometry, an orthonormal matrix ${\mathbf{R}}$ is used to link ${{\mathbf{X}}}$ and ${{\mathbf{X}}}_b$ through 
\begin{equation}
{{\mathbf{X}}} = {{\mathbf{R}}}{{\mathbf{X}_b}}, \quad {\mathbf{R}} \in {\mathbb{O}^{3 \times q}},
\label{eq:link1}
\end{equation}
where ${\mathbf{R}}^{\text{T}}{\mathbf{R}} = {\mathbf{I}_q}$, and ${\mathbf{I}_q}$ represents an identity matrix of size $q$. The matrix $\mathbf{R}$ indicates the orientation of the platform's body frame relative to the reference coordinate. The combination of \eref{eq:model1} and \eref{eq:link1} yields the orthonormality-constrained attitude model (OC-AM).

The OC-AM involves two kinds of constraints, namely, the integer constraints for the ambiguities and the orthonormality constraints for matrix ${\mathbf{R}}$. The optimization for the OC-AM can be expressed as \begin{equation}
\mathop {\min }\limits_{{\mathbf{R}} \in {\mathbb{O}^{3 \times q}}, {\mathbf{N}} \in {\mathbb{Z}^{\mathcal{S} \times\mathcal{A}}}}\left\| {\operatorname{vec}\left( {{\mathbf{Y}} - {\mathbf{AR}}{\mathbf{X}_b} - {\mathbf{BN}}} \right)} \right\|_{\mathbf{Q}_{\mathbf{Y\!Y}}^{-1}}^2,
\label{eq:ocam}
\end{equation}
which can be decomposed into three terms \cite{giorgi2012instantaneous, 6491499}
\begin{equation}
\resizebox{.88\hsize}{!}{
$
\begin{aligned}
&\mathop {\min }\limits_{{\mathbf{R}} \in {\mathbb{O}^{3 \times \mathcal{A}}}, {\mathbf{N}} \in {\mathbb{Z}^{\mathcal{S} \times\mathcal{A}}}} \!\left\| {\operatorname{vec} \! \left( {{\mathbf{Y}} \!-\! {\mathbf{AR}}{\mathbf{X}_b} \!-\! {\mathbf{BN}}} \right)}\! \right\|_{\mathbf{Q}_{\mathbf{Y\!Y}}^{-1}}^2
= \left\| {\operatorname{vec} \! \left( {\mathbf{E}}_{_\text{AC}} \right)}\! \right\|_{\mathbf{Q}_{\mathbf{Y\!Y}}^{-1}}^2 \\ &+ \!\mathop {\min }\limits_{{\mathbf{N}} \in {\mathbb{Z}^{\mathcal{S} \times\mathcal{A}}}} \!\! \left(\!\! \left\| {\operatorname{vec} \!\! \left( \! {\mathbf{N}} \! - \! {\mathbf{\hat N}_{_\text{AC}}} \!\right)} \!\right\|_{\mathbf{Q}_{{\mathbf{\hat N}_{_\text{AC}}\mathbf{\hat N}_{_\text{AC}}}}^{-1}}^2\!\!\! +\!
\mathop {\min }\limits_{{\mathbf{R}} \in {\mathbb{O}^{3 \times q}}}\!\!\left\| {\operatorname{vec} \!\! \left(\!{\mathbf{R}} \!-\! {\mathbf{\hat R}_{_\text{AC}}}\!\!\left(\!{\mathbf{N}}\!\right) \!\right)} \!\right\|_{\mathbf{Q}_{{\mathbf{\hat R}_{_\text{AC}}}\!\left(\!{\mathbf{N}}\!\right){\mathbf{\hat R}_{_\text{AC}}}\!\left(\!{\mathbf{N}}\!\right)}^{-1}}^2 \!\right)\!.
\label{eq:decom3}
\end{aligned}
$}
\end{equation}
Both \eref{eq:decom2} and \eref{eq:decom3} use the same float solution. The only difference between \eref{eq:decom2} and \eref{eq:decom3} is the solution space for the matrix $\mathbf{R}$. Due to the orthonormality constraint, one has to take the last term of \eref{eq:decom3} into account. For the OC-AM, the optimization in \eref{eq:ocam} is equivalent to \cite{giorgi2012instantaneous}
\begin{equation}
\resizebox{.85\hsize}{!}{
$
\begin{aligned}
\mathop {\min }\limits_{{\mathbf{N}} \in {\mathbb{Z}^{\mathcal{S} \times\mathcal{A}}}} \!\! \left(\!\! \left\| {\operatorname{vec} \! \left( \! {\mathbf{N}} \! - \! {\mathbf{\hat N}_{_\text{AC}}} \!\right)} \!\right\|_{\mathbf{Q}_{{\mathbf{\hat N}_{_\text{AC}}\!\mathbf{\hat N}_{_\text{AC}}}}^{-1}}^2\!\!\! +\!
\mathop {\min }\limits_{{\mathbf{R}} \in {\mathbb{O}^{3 \times q}}}\!\left\| {\operatorname{vec} \! \left(\!{\mathbf{R}} \!-\! {\mathbf{\hat R}_{_\text{AC}}}\!\!\left(\!{\mathbf{N}}\!\right) \!\right)} \!\right\|_{\mathbf{Q}_{{\mathbf{\hat R}_{_\text{AC}}}\!\left(\!{\mathbf{N}}\!\right){\mathbf{\hat R}_{_\text{AC}}}\!\left(\!{\mathbf{N}}\!\right)}^{-1}}^2 \!\right).
\label{eq:lambda3}
\end{aligned}
$}
\end{equation}
The OC-AM is more robust than the UC-AM and AC-AM. Nowadays, the MC-LAMBDA method is the most notable approach to solve (\ref{eq:lambda3}), which takes advantage of the search-and-shrink or search-and-expand algorithms to speed up the integer search procedures \cite{giorgi2010carrier}. However, the search space of \eref{eq:lambda3} is no longer an ellipsoidal set, which leads to more complexity than that of \eref{eq:lambda1} and \eref{eq:lambda2}. It is challenging to solve the complicated non-convex optimization \eref{eq:lambda3}, especially to satisfy the requirements of real-time applications.
Therefore, it is essential to develop a more efficient method to alleviate the high computational burden of ambiguity resolution encountered in many practical situations.

\section{Overview of Riemannian Manifold Optimization} 
\label{sec:ove}
In this work, we formulate the GNSS attitude determination problem as an optimization on Riemannian manifolds. Before studying the details of the proposed method, we first briefly introduce all of the necessary ingredients for Riemannian algorithms on matrix manifolds, including the manifold terminology, the Riemannian derivatives, and the optimization algorithm design. 

\subsection{Manifold Optimization Terminology}
A matrix manifold $\mathcal{M}$ refers to a topological space that is in bijection with an open space of the Euclidean space. At each point $\mathbf{X} \in \mathcal{M}$, the manifold $\mathcal{M}$ locally resembles a linear space, i.e., a Euclidean space, which is known as the tangent space $\mathcal{T}_{_\mathbf{X}} \mathcal{M}$. The dimension of $\mathcal{T}_{_\mathbf{X}} \mathcal{M}$ represents the dimension of $\mathcal{M}$, namely, the degrees of freedom.

The tangent space $\mathcal{T}_{_\mathbf{X}} \mathcal{M}$ can be endowed with a positive-definite inner product $\langle .,. \rangle_{\mathbf{X}}$ that is called the Riemannian metric, which makes it possible to define the notion of length for tangent vectors. Although there are various Riemannian metrics for a manifold $\mathcal{M}$, a typical option is to employ the inner product in canonical form, permitting the simple expressions of the Riemannian gradient and Hessian.

Let $f: \mathcal{M} \longrightarrow \mathbb{R}$ be a smooth function over a matrix manifold $\mathcal{M}$. The directional derivative of $f$ at the point $\mathbf{X} \in \mathcal{M}$ in the direction $\bm{\xi}_{_\mathbf{X}} \in \mathcal{T}_{_\mathbf{X}} \mathcal{M}$, expressed by $\text{D}(f(\mathbf{X}))[\bm{\xi}_{_\mathbf{X}}]$, is defined as
\begin{align}
\text{D}(f(\mathbf{X}))[\bm{\xi}_{_\mathbf{X}}] = \lim_{t \rightarrow 0} \cfrac{f(\mathbf{X}+t\bm{\xi}_{_\mathbf{X}}) - f(\mathbf{X})}{t}. 
\label{eq:direcderiv}
\end{align}
Given that the tangent space $\mathcal{T}_{_\mathbf{X}}\mathcal{M}$ is a linear approximation of the Riemannian manifold $\mathcal{M}$ near the point $\mathbf{X}$, only the tangent vectors $\bm{\xi}_{_\mathbf{X}} \in \mathcal{T}_{_\mathbf{X}}\mathcal{M}$ can be utilized as valid direction vectors. The operator $\text{D}(f(\mathbf{X})): \mathcal{T}_{_\mathbf{X}}\mathcal{M} \longrightarrow \mathbb{R}$ is defined as the indefinite directional derivative of $f$ at $\mathbf{X}$, which gives the the directional derivative $\text{D}(f(\mathbf{X}))[\bm{\xi}_{_\mathbf{X}}]$ for a specific tangent vector $\bm{\xi}_{_\mathbf{X}}$.

\subsection{Riemannian Optimization Algorithms}
To design optimization algorithms over Riemannian manifolds, one needs the notions of Riemannian gradient $\overline \nabla_{_\mathbf{X}} f$ and Riemannian Hessian $\overline \nabla_{_\mathbf{X}}^2 f$, which are derivative operators defined only on the tangent space $\mathcal{T}_{_\mathbf{X}}\mathcal{M}$, unlike Euclidean gradient $\nabla_{_\mathbf{X}} f$ and Euclidean Hessian $\nabla_{_\mathbf{X}}^2 f$ that are valid in Euclidean space and not exclusively on $\mathcal{T}_{_\mathbf{X}}\mathcal{M}$. At the point $\mathbf{X} \in \mathcal{M}$, the Euclidean gradient $\nabla_{_\mathbf{X}} f$ is related to the directional derivative $\bm{\xi}_{_\mathbf{X}} \in \mathcal{T}_{_\mathbf{X}} \mathcal{M}$ as
\begin{align}
\langle \nabla_{_\mathbf{X}} f, \bm{\xi}_{_\mathbf{X}} \rangle_{_\mathbf{X}} = \text{D}(f(\mathbf{X}))[\bm{\xi}_{_\mathbf{X}}].
\label{eq:direc}
\end{align}
When the canonical Riemannian metric is used, one can easily achieve the Riemannian gradient $\overline \nabla_{_\mathbf{X}} f$ via projecting the Euclidean gradient $\nabla_{_\mathbf{X}} f$ onto the tangent space, that is
\begin{align}
\overline \nabla_{_\mathbf{X}} f = \Pi_{_\mathbf{X}} \! \left( \nabla_{_\mathbf{X}} f\right),
\label{eq:grad}
\end{align}
where $\Pi_{\mathbf{x}}$ denotes the corresponding orthogonal projection. Similar to the definition of Euclidean Hessian, one can define the Riemannian Hessian $\overline \nabla_{_\mathbf{X}}^2 f$ as the directional derivative of the Riemannian gradient. However, the directional derivative of $\overline \nabla_{_\mathbf{X}} f$ does not have to lie in the tangent space $\mathcal{T}_{_\mathbf{X}}\mathcal{M}$ such that further operation is required to project the directional derivative of $\overline \nabla_{_\mathbf{X}} f$ onto the tangent space, i.e.,
\begin{align}
\overline \nabla_{_\mathbf{X}}^2 f[\bm{\xi}_{\mathbf{x}}] = \Pi_{\mathbf{x}}(\text{D}(\overline \nabla_{_\mathbf{X}} f)[\bm{\xi}_{\mathbf{x}}]),
\label{eq:hess}
\end{align}
in which the same orthogonal projection $\Pi_{_\mathbf{X}}$ is applied.

Rather than the nonlinear constrained optimization in Euclidean space, Riemannian optimization generally involves unconstrained optimization on Riemannian manifold, i.e., unconstrained optimization over a constrained set. Therefore, Riemannian optimization algorithms follow similar procedures as the unconstrained ones in Euclidean space, but some discrepancies need to be addressed. Initially, one needs to locally approximate the manifold around the point $\mathbf{X}$ on the given manifold to achieve a linear space. In other words, the first step is to derive the formulation of the tangent space $\mathcal{T}_{_\mathbf{X}}\mathcal{M}$ and endow it with the Riemannian metric. Then, based on the Riemannian gradient and Hessian, we can obtain a descent direction $\bm{\xi}_{_\mathbf{X}} \in \mathcal{T}_{_\mathbf{X}} \mathcal{M}$ and a step size $\alpha$ so that we can find a new point $\mathbf{X}+\alpha\bm{\xi}_{_\mathbf{X}}$ contained in the tangent space. However, the newly found point is not a feasible solution since it is not on the manifold. Finally, a retraction operator $R_{_\mathbf{X}}$ is utilized to project the new point to the manifold.

The procedures of Riemannian optimization are summarized in \algref{alg1}. The Riemannian algorithms that rely on only the gradient information refer to the first-order algorithms, where the steepest descent direction can be chosen as $\bm{\xi}_{_\mathbf{X}} = - \frac{\overline \nabla_{_\mathbf{X}} f}{||\overline \nabla_{_\mathbf{X}} f||_{_{\mathbf{X}}}}$. In contrast, those who take advantage of the Riemannian Hessian belong to the second-order algorithms, among which, one example is Newton’s method on Riemannian manifolds that figures out the search direction by solving
$\overline \nabla_{_\mathbf{X}}^2 f[\bm{\xi}_{\mathbf{x}}] = -\overline \nabla_{_\mathbf{X}} f$.

According to the discussion above, the Riemannian optimization algorithms require a few ingredients, namely, the tangent space, the orthogonal projector, the Riemannian gradient, the Riemannian Hessian, and the retraction operator. For more details about Riemannian optimization, refer to \cite{smith1994optimization, absil2009optimization,liu2020simple}. 


\begin{algorithm}[t!]
\begin{algorithmic}[1]
\STATE Initialize ${\mathbf{X}} \in \mathcal{M}$.
\WHILE {$||\overline \nabla_{_\mathbf{X}} f||_{_{\mathbf{X}}} \neq 0$}
\STATE Find the search direction $\bm{\xi}_{_\mathbf{X}} \in \mathcal{T}_{_\mathbf{X}}\mathcal{M}$ using $\overline \nabla_{_\mathbf{X}} f$ and/or $\overline \nabla_{_\mathbf{X}}^2 f$.
\STATE Compute the step size $\alpha$ using backtracking.
\STATE Retract ${\mathbf{X}} = R_{_\mathbf{X}}( \alpha \bm{\xi}_{_\mathbf{X}})$.
\ENDWHILE
\end{algorithmic}
\caption{Procedures of Riemannian Optimization.}
\label{alg1}
\end{algorithm}

\section{The Proposed Riemannian-manifold-based Attitude Determination Method}
\label{sec:proposed}
For the integer search-based methods discussed in \sref{sec: model}, ambiguity resolution is carried out via searching around the float solution in the integer domain; that is, the float solution acts as an intermediate to simplify the integer search. As a consequence, the performance of ambiguity resolution profoundly relies on the quality of the float solution. Hence, it is critical to ensure that the float solution is of sufficient accuracy. Indeed, the more accurate the float solution, the better. To facilitate resolving the carrier-phase ambiguities more efficiently and reliably, we propose an approach based on optimization techniques on Riemannian manifolds to enhance the float solution. Subsequently, we develop an efficient integer search algorithm based on the improved float solution. The proposed method employs the OC-AM with \eref{eq:ocam} being the objective function, which is named as Riemannian-manifold-based orthonormality-constrained attitude
determination (RieMOCAD).

\subsection{The Intermediate Float Solution And Riemannian Algorithm Design}
\label{sec: IntSol}
In this section, we first formulate the float solution used in this work. Then, we study and characterize the geometry of the manifolds of interest and present Riemannian optimization algorithms employed to achieve the float solution.

\subsubsection{The Float Solution}
The key principle behind the proposed method is to maintain the orthonormality constraint ${\mathbf{R}} \in {\mathbb{O}^{3 \times q}}$ when pursuing the initial float solution. The motivation here is that preserving this constraint is expected to provide a better float solution (than the one obtained with the constraint ignored). Incorporating the constrained, however, usually leads to an increased computational complexity. To develop a computationally efficient solution, we leverage advanced tools from \emph{manifold optimization} \cite{smith1994optimization, absil2009optimization,liu2020simple}. 

Due to the orthonormality characteristics, the solution set of ${\mathbf{R}}$ is a manifold, which allows us to solve the problem using efficient Riemannian algorithms. In other words, we calculate the float solution by using Riemannian manifold optimization, which gives us
\begin{equation}
\resizebox{.88\hsize}{!}{
$
{{{\mathbf{\hat R}}}_{_\text{RM}} }, {{{\mathbf{\hat N}}}_{_\text{RM}} }\! \! =\! \!\mathop {\arg \min }\limits_{{\mathbf{R}} \in {\mathbb{O}^{3 \times q}}, {\mathbf{N}} \in {\mathbb{R}^{\mathcal{S} \times\mathcal{A}}}} \! \left\| \! {\operatorname{vec}\! \left( {{\mathbf{Y}} \! -\!  {\mathbf{AR}}{\mathbf{X}_b} \! - \! {\mathbf{BN}}} \right)}\!  \right\|_{\mathbf{Q}_{\mathbf{Y\!Y}}^{-1}}^2.
$}
\label{eq:floatrm}
\end{equation}

The incorporation of the orthonormality constraint ${\mathbf{R}} \in {\mathbb{O}^{3 \times q}}$ can improve the estimation of ${\mathbf{R}}$ using \eref{eq:floatrm} compared to \eref{eq:float2}. We can generally expect ${{{\mathbf{\hat R}}}_{_\text{RM}}}$ to be more accurate than ${{{\mathbf{\hat R}}}_{_\text{AC}}}$. Given the convexity of the objective functions \eref{eq:float2} and \eqref{eq:floatrm} in the unconstrained ambiguities, the float solutions obtained from these optimizations should satisfy
\begin{equation}
{{{\mathbf{\hat N}}}_{_\text{AC}} }\! \! =\! \left(  {\mathbf{B}}^{\text{T}} {\mathbf{Q}_{\mathbf{Y\!Y}}^{-1}}{\mathbf{B}}\right){\mathbf{B}}^{\text{T}} {\mathbf{Q}_{\mathbf{Y\!Y}}^{-1}}\!\operatorname{vec} \!\!\left( \!{\mathbf{Y}} \! -\!  {\mathbf{A}}{{{\mathbf{\hat R}}}_{_\text{AC}}}{\mathbf{X}_b}\!\right),
\label{eq:floatrm1}
\end{equation}
\begin{equation}
{{{\mathbf{\hat N}}}_{_\text{RM}} }\! \! =\! \left(  {\mathbf{B}}^{\text{T}} {\mathbf{Q}_{\mathbf{Y\!Y}}^{-1}}{\mathbf{B}}\right){\mathbf{B}}^{\text{T}} {\mathbf{Q}_{\mathbf{Y\!Y}}^{-1}}\!\operatorname{vec} \!\!\left( \!{\mathbf{Y}} \! -\!  {\mathbf{A}}{{{\mathbf{\hat R}}}_{_\text{RM}}}{\mathbf{X}_b}\!\right).
\label{eq:floatrm2}
\end{equation}
Given that ${{{\mathbf{\hat R}}}_{_\text{RM}} }$ is less erroneous than ${{{\mathbf{\hat R}}}_{_\text{AC}} }$, we can expect ${{{\mathbf{\hat N}}}_{_\text{RM}} }$ to be of higher quality than ${{{\mathbf{\hat N}}}_{_\text{AC}} }$.

Note that the LS solutions ${{{\mathbf{\hat R}}}_{_\text{AC}} }$ and ${{{\mathbf{\hat N}}}_{_\text{AC}} }$ are also required in the integer search process (presented later). Hence, the first step is to compute the LS solution ${{{\mathbf{\hat R}}}_{_\text{AC}} }$ and ${{{\mathbf{\hat N}}}_{_\text{AC}} }$, which allows us to obtain ${{{\mathbf{\hat R}}}_{_\text{RM}} }$ and ${{{\mathbf{\hat N}}}_{_\text{RM}} }$ based on ${{{\mathbf{\hat R}}}_{_\text{AC}} }$ and ${{{\mathbf{\hat N}}}_{_\text{AC}} }$ rather than solving \eref{eq:floatrm} directly. To determine the solution of \eref{eq:floatrm}, we reformulate the objective function as a sum of several squares terms. This reformulation helps us in exploiting Riemannian optimization with an easy-compute form.


\begin{lemma} \label{le:decomfloat}
Let ${{{\mathbf{\hat R}}}_{_\text{AC}} }$ and ${{{\mathbf{\hat N}}}_{_\text{AC}} }$ be the LS solutions given in \eref{eq:float2}. Then, we have
\begin{equation}
\resizebox{.88\hsize}{!}{
$
\begin{aligned}
&\left\| {\operatorname{vec}\!\left( \!{{\mathbf{Y}}\! -\!{\mathbf{AR}}{\mathbf{X}_b} \!- \!{\mathbf{BN}}}\! \right)}\! \right\|_{\mathbf{Q}_{\mathbf{Y\!Y}}^{-1}}^2 = \left\| {\operatorname{vec} \! \left( {\mathbf{E}_{_\text{AC}}} \right)}\! \right\|_{\mathbf{Q}_{\mathbf{Y\!Y}}^{-1}}^2\\
& +\! \left\| \mathrm{vec} \!\left ( \mathbf{R} \!-\!\mathbf{\hat R}_{_\text{AC}}\!\right )\! \right\|_{\mathbf{Q}_{{\mathbf{\hat R}_{_\text{AC}}}\!{\mathbf{\hat R}_{_\text{AC}}}}^{-1}}^2 \!+\! \left\| \mathrm{vec} \!\left ( \mathbf{N} \!-\!\mathbf{\hat N}_{_\text{AC}}\!\!\left (\mathbf{R} \right )\!\right )\! \right\|_{\mathbf{Q}_{{\mathbf{\hat N}_{_\text{AC}}}\!\left({\mathbf{R}}\right){\mathbf{\hat N}_{_\text{AC}}}\!\left({\mathbf{R}}\right)}^{-1}}^2,
\end{aligned}
$}
\label{eq:lemf}
\end{equation}
with
\begin{equation}
\resizebox{.88\hsize}{!}{
$
\begin{aligned}
&\operatorname{vec} \!\!  \left(\!\!{\mathbf{\hat N}_{_\text{AC}}}\!\!\left(\!{\mathbf{R}}\!\right)\!\! \right) \!= \!\operatorname{vec}\!\left( \!{\mathbf{\hat N}_{_\text{AC}}}\! \right) \!-\! \mathbf{Q}_{_{{\mathbf{\hat N}_{_\text{AC}}}\!{\mathbf{\hat R}_{_\text{AC}}}}}\!\! \mathbf{Q}_{_{{\mathbf{\hat R}_{_{\text{AC}}}}\!{\mathbf{\hat R}_{_\text{AC}}}}}^{-1}\!\!\operatorname{vec} \!\!  \left(\!\!\mathbf{\hat R}_{_{\text{AC}}}\! \!- \!\mathbf{R}\!\!\right)\!.
\end{aligned}
$}
\label{eq:rm1cond1}
\end{equation}
\end{lemma}
\begin{proof} 
See Appendix A.
\end{proof}

According to \lref{le:decomfloat}, the optimization problem in \eref{eq:floatrm} is equivalent to
\begin{equation}
{{{\mathbf{\hat R}}}_{_\text{RM}} } \!=\mathop {\arg \min }\limits_{{\mathbf{R}} \in {\mathbb{O}^{3 \times q}}} \! \left\| \mathrm{vec} \!\!\left ( \mathbf{R} \!-\!\mathbf{\hat R}_{_\text{AC}}\!\right )\! \right\|_{\mathbf{Q}_{{\mathbf{\hat R}_{_\text{AC}}}\!{\mathbf{\hat R}_{_\text{AC}}}}^{-1}}^2,
\label{eq:rm1}
\end{equation}
\begin{equation}
\begin{aligned}
{{{\mathbf{\hat N}}}_{_\text{RM}} }\! \! =\! \! {\mathbf{\hat N}_{_\text{AC}}}\!\!\left(\!{\mathbf{\hat R}_{_\text{RM}}}\!\right).
\end{aligned}
\label{eq:rm1cond}
\end{equation}
That is to say, solving \eref{eq:rm1} is where Riemannian optimization comes into play.

\subsubsection{Optimization Algorithm Design on Riemannian Manifolds}
\label{sec:RieOpt}

As stated in \sref{sec: ocmodel}, the attitude matrix $\mathbf{R}$ is an orthonormal matrix whose columns are orthogonal regarding the inner product. Here, for the sake of simplicity and avoiding confusion, we still utilize the symbol $\mathbf{X}$ to denote the point on the manifold rather than symbol $\mathbf{R}$. Then, we can use the terminologies discussed in \sref{sec:ove} and present the required elements for Riemannian algorithms.

The potential solution set $\mathcal{M}$ it an embedded submanifold of $\mathbb{R}^{3 \times q}$ ($q = 1, 2, 3$), which is well known as the Stiefel manifold given by \cite{boumal2020introduction}
\begin{align}
\mathcal{M} = \left\{\mathbf{X} \in \mathbb{R}^{3 \times q} \ | \ \mathbf{X}^{\text{T}}\mathbf{X}=\mathbf{I}_q\right\}.
\end{align}
The dimension of the the Stiefel manifold $\mathcal{M}$ is 
\begin{align}
{\operatorname{dim}} \mathcal{M} \!=\! {\operatorname{dim}} \mathbb{R}^{3 \times q} \!-\! {\operatorname{dim}} {\operatorname{Sym}}\!\left(q\right) \!=\! 3q\! -\! \frac{q\left(q+1\right)}{2},
\end{align}
where ${\operatorname{Sym}}\!\left(q\right)$ represents the linear space of symmetric matrices of size $q$.

The tangent spaces of the the Stiefel manifold $\mathcal{M}$, a subspace of $\mathbb{R}^{3 \times q}$, is given by
\begin{align}
\mathcal{T}_{_\mathbf{X}}\mathcal{M} = \left\{\mathbf{V} \in \mathbb{R}^{3 \times q} \!\ | \ \!\mathbf{X}^{\text{T}}\mathbf{V} + \mathbf{V}^{\text{T}}\mathbf{X}=\mathbf{0}\right\}.
\end{align}
Alternatively, one can also express the tangent vectors in an explicit form as
\begin{align}
\mathcal{T}_{_\mathbf{X}}\mathcal{M} \!= \!\left\{\!\mathbf{X} \mathbf{S}\! +\! \mathbf{X}_{\perp} \mathbf{K} \!\! \ | \ \!\!\mathbf{S} \!\in\! {\operatorname{Skew}}\!\left(q\right)\!, \mathbf{K} \!\in\! \mathbb{R}^{(3-q) \times q} \!\right\}\!,
\end{align}
where $\mathbf{X}$ and $\mathbf{X}_{\perp}$ constitute the orthonormal basis of $\mathbb{R}^{3 \times 3}$, and 
\begin{align}
{\operatorname{Skew}}\!\left(q\right) = \left\{\mathbf{S} \in \mathbb{R}^{q \times q} \ | \ \mathbf{S}^{\text{T}} = -\mathbf{S}\right\}
\end{align}
denotes the set of skew-symmetric matrices of size $q$.

At a point $\mathbf{X}$ on the Stiefel manifold $\mathcal{M}$, the orthogonal projection to the tangent space $\mathcal{T}_{_\mathbf{X}}\mathcal{M}$ can be formulated as
\begin{equation}
\begin{aligned}
\Pi_{_\mathbf{X}} \! \left( \mathbf{U}\right) = &\mathbf{U} - \mathbf{X}\frac{\mathbf{X}^{\text{T}}\mathbf{U}+\mathbf{U}^{\text{T}}\mathbf{X}}{2}\\
=&\left( \mathbf{I} - \mathbf{X}\mathbf{X}^{\text{T}}\right)\!\mathbf{U} +  \mathbf{X}\frac{\mathbf{X}^{\text{T}}\mathbf{U}-\mathbf{U}^{\text{T}}\mathbf{X}}{2}.
\end{aligned}
\end{equation}
Note that the orthogonal projector $\Pi_{_\mathbf{X}}$ guarantees that $\mathbf{U} - \Pi_{_\mathbf{X}} \! \left( \mathbf{U}\right)$ falls into the normal space $\mathcal{T}^{\perp}_{_\mathbf{X}}\mathcal{M}$ of the tangent space $\mathcal{T}_{_\mathbf{X}}\mathcal{M}$. According to \eref{eq:direc}, \eref{eq:grad} and \eref{eq:hess}, we can compute the Riemannian gradient $\overline \nabla_{_\mathbf{X}} f$ and Riemannian Hessian $\overline \nabla_{_\mathbf{X}}^2 f$. Then, we obtain
\begin{align}
\overline \nabla_{_\mathbf{X}} f = \nabla_{_\mathbf{X}} f -  \mathbf{X}\frac{\mathbf{X}^{\text{T}}{\nabla_{_\mathbf{X}} f}+\left( \nabla_{_\mathbf{X}} f\right)^{\text{T}}\mathbf{X}}{2},
\end{align}
\begin{align}
\overline \nabla_{_\mathbf{X}}^2 f[\bm{\xi}_{\mathbf{x}}]\! =\! \Pi_{\mathbf{x}}(\nabla_{_\mathbf{X}}^2 f[\bm{\xi}_{\mathbf{x}}])\!-\! \bm{\xi}_{\mathbf{x}}\!\frac{\mathbf{X}^{\text{T}}{\nabla_{_\mathbf{X}} f}\!+\!\left( \nabla_{_\mathbf{X}} f\right)^{\text{T}}\!\mathbf{X}}{2}.
\end{align}

After finding a new point in $\mathcal{T}_{_\mathbf{X}}\mathcal{M}$ using $\overline \nabla_{_\mathbf{X}} f$ and/or $\overline \nabla_{_\mathbf{X}}^2 f$, retraction is required to turn the new point into the manifold $\mathcal{M}$. For the the Stiefel manifold, multiple retraction operators exist. For instance, the Q-factor retraction is an efficient operator to retract the updated point to the manifold \cite{boumal2020introduction}:
\begin{align}
R_{_\mathbf{X}}({\mathbf{V}}) = \mathbf{Q},
\end{align}
in which the QR decomposition is employed such that $\mathbf{QR}_{ut}=\mathbf{X}+\mathbf{V}$ with $\mathbf{Q} \in \mathcal{M}$ and $\mathbf{R}_{ut} \in \mathbb{R}^{q \times q}$. Here, $\mathbf{R}_{ut}$ represents an upper triangular with nonnegative diagonal entries. On the other hand, the polar retraction can also be used \cite{boumal2020introduction}, which is expressed as
\begin{equation}
\begin{aligned}
R_{_\mathbf{X}}({\mathbf{V}}) = &\left( \mathbf{X}+\mathbf{V}\right)\left( \left( \mathbf{X}+\mathbf{V}\right)^\text{T}\left( \mathbf{X}+\mathbf{V}\right)\right)^{-\frac{1}{2}}\\
=&\left( \mathbf{X}+\mathbf{V}\right)\left(\mathbf{I}_q +\mathbf{V}^\text{T}\mathbf{V}\right)^{-\frac{1}{2}},
\end{aligned}
\end{equation}
where $\left( \cdot \right)^{-\frac{1}{2}}$ represents the inverse matrix square root. For both retraction operators, it readily seen that $R_{_\mathbf{X}}({\mathbf{0}}) = \mathbf{X}$. In this work, we employ the polar retraction.

\subsection{Ambiguity Resolution Aided by the Float Solution}
\label{sec:ambres}
\subsubsection{Decomposition of the Objective Function}
\label{sec:decompose}
As stated in \sref{sec: model}, decomposing the objective function to a few more simple terms plays a crucial part in ambiguity resolution in the sense of defining the search space and simplifying the search process. However, the orthogonal decomposition \eref{eq:decom1} and \eref{eq:decom2} rely on an essential fact that the first-order and high-order (higher than two) derivatives of the objective function at the LS solution are zero \cite{6164274}. There is no doubt that the orthogonal decomposition is not feasible at any other point except for the point where the gradient vanishes corresponding to the LS solution, prompting us to develop a decomposition method to use the float solution ${{{\mathbf{\hat R}}}_{_\text{RM}} }$ and ${{{\mathbf{\hat N}}}_{_\text{RM}} }$. 


In order to utilize the high-quality float solution ${{{\mathbf{\hat R}}}_{_\text{RM}}}$ and ${{{\mathbf{\hat N}}}_{_\text{RM}}}$, we first consider decomposing \eref{eq:ocam} at an arbitrary point, as shown in \lref{th:decom}.
\begin{lemma} \label{th:decom}
For $\forall {\mathbf{\bar R}} \in {\mathbb{R}^{3 \times q}}$ and $\forall \mathbf{\bar N} \in {\mathbb{R}^{\mathcal{S} \times \mathcal{A}}}$, we can rewrite the objective function of \eref{eq:ocam} as a sum of five terms expressed by ${{{\mathbf{\bar R}}}}$ and ${{\mathbf{\bar N}}}$ as
\begin{equation}
\resizebox{.88\hsize}{!}{
$
\begin{aligned}
&\left\| {\operatorname{vec}\!\left( \!{{\mathbf{Y}}\! -\!{\mathbf{AR}}{\mathbf{X}_b} \!- \!{\mathbf{BN}}}\! \right)}\! \right\|_{\mathbf{Q}_{\mathbf{Y\!Y}}^{-1}}^2 = \left\| {\operatorname{vec} \! \left( {\mathbf{\bar E}} \right)}\! \right\|_{\mathbf{Q}_{\mathbf{Y\!Y}}^{-1}}^2\\
& +\! \left\|\! \mathrm{vec} \!\left ( \!\mathbf{N} \!-\!\mathbf{\bar N}\!\right )\! \right\|_{\mathbf{Q}_{{\mathbf{\hat N}_{_\text{AC}}\!\mathbf{\hat N}_{_\text{AC}}}}^{-1}}^2\!\! \!+\! \left\| \mathrm{vec} \!\left ( \mathbf{R} \!-\!\mathbf{\bar R}\!\left (\!\mathbf{N}\! \right )\!\right )\! \right\|_{\mathbf{Q}_{{\mathbf{\hat R}_{_\text{AC}}}\!\left({\mathbf{N}}\right){\mathbf{\hat R}_{_\text{AC}}}\!\left({\mathbf{N}}\right)}^{-1}}^2\\
&+ 2\mathrm{vec}\! \left ( \mathbf{N} \!-\!\mathbf{\bar N}\right )^{\text{T}}{\mathbf{Q}_{{\mathbf{\hat N}_{_\text{AC}}\!\mathbf{\hat N}_{_\text{AC}}}}^{-1}}\!\mathrm{vec}\! \left ( \!\mathbf{\bar N} \!-\!\mathbf{\hat N}_{_\text{AC}}\!\right ) \\
&+ 2\mathrm{vec} \!\left ( \!\mathbf{R} \!-\!\mathbf{\bar R}\!\left (\!\mathbf{N} \!\right )\!\right )^{\text{T}}\!{\mathbf{Q}_{{\mathbf{\hat R}_{_\text{AC}}}\!\left({\mathbf{N}}\right){\mathbf{\hat R}_{_\text{AC}}}\!\left({\mathbf{N}}\right)}^{-1}}\!\mathrm{vec} \!\left ( \!\mathbf{\bar R} \!-\!\mathbf{\hat R}_{_\text{AC}}\!\!\left (\!\mathbf{\bar N} \!\right )\!\right )
\end{aligned}
$}
\label{eq:th1}
\end{equation}
with ${\mathbf{\bar E}} = {{\mathbf{Y}} \!-\! {\mathbf{A}}{\mathbf{\bar R}}{\mathbf{X}_b}  \!-\! {\mathbf{B \bar N}}}$.
\end{lemma}
\begin{proof} 
See Appendix B.
\end{proof}

Compared with \eref{eq:decom3}, \eref{eq:th1} contains five terms on the right-hand side, among which the first represents the squares norm of the residual corresponding to ${{{\mathbf{\bar R}}}}$ and ${{\mathbf{\bar N}}}$, the second measures the the distance from the integer matrix ${{\mathbf{N}}}$ to ${{\mathbf{\bar N}}}$ in the metric of ${\mathbf{Q}_{{\mathbf{\hat N}_{_\text{AC}}\!\mathbf{\hat N}_{_\text{AC}}}}^{-1}}$, the third weighs the distance from the orthonormal matrix ${{{\mathbf{R}}}}$ to $\mathbf{\bar R}\!\left (\!\mathbf{N}\! \right )$ in the metric of ${\mathbf{Q}_{{\mathbf{\hat R}_{_\text{AC}}}\!\left({\mathbf{N}}\right){\mathbf{\hat R}_{_\text{AC}}}\!\left({\mathbf{N}}\right)}^{-1}}$, and the last two terms are related to distance from the selected points ${{{\mathbf{\bar R}}}}$ and ${{\mathbf{\bar N}}}$ to the LS solution $\mathbf{\hat R}_{_\text{AC}}$ and $\mathbf{\hat N}_{_\text{AC}}$. Note that the covariance matrixes of the LS solution are used in \eref{eq:th1}, and we can always compute these matrixes using a close form. If ${\mathbf{\bar R}} =\mathbf{\hat R}_{_\text{AC}}$ and $\mathbf{\bar N} =\mathbf{\hat N}_{_\text{AC}}$, the last two terms in \eref{eq:th1} equal to zero, leading to the identical expression as \eref{eq:decom3}.

\fref{fig:decomp} shows a geometric illustration of the decomposition of the objective function. The essence of the optimization \eref{eq:ocam} is to seek for the solution that minimizes $d^2$ in \fref{fig:decomp}. Decomposing $d^2$ into a few easy-evaluated terms allows the full usage of the float solutions as an intermediate to simplify the search process. This figure depicts one of the most simple cases and indicates the difference between the two decomposition methods. 



To simplify ambiguity resolution, we try to reformulate \eref{eq:th1} as an alternative form that is easy to be bounded and evaluate.

\begin{figure*}[htbp]
\centering
\subfigure[Orthogonal decomposition \cite{teunissen2012gps}] { \label{fig:decomp1a} 
\includegraphics[width=0.98\columnwidth]{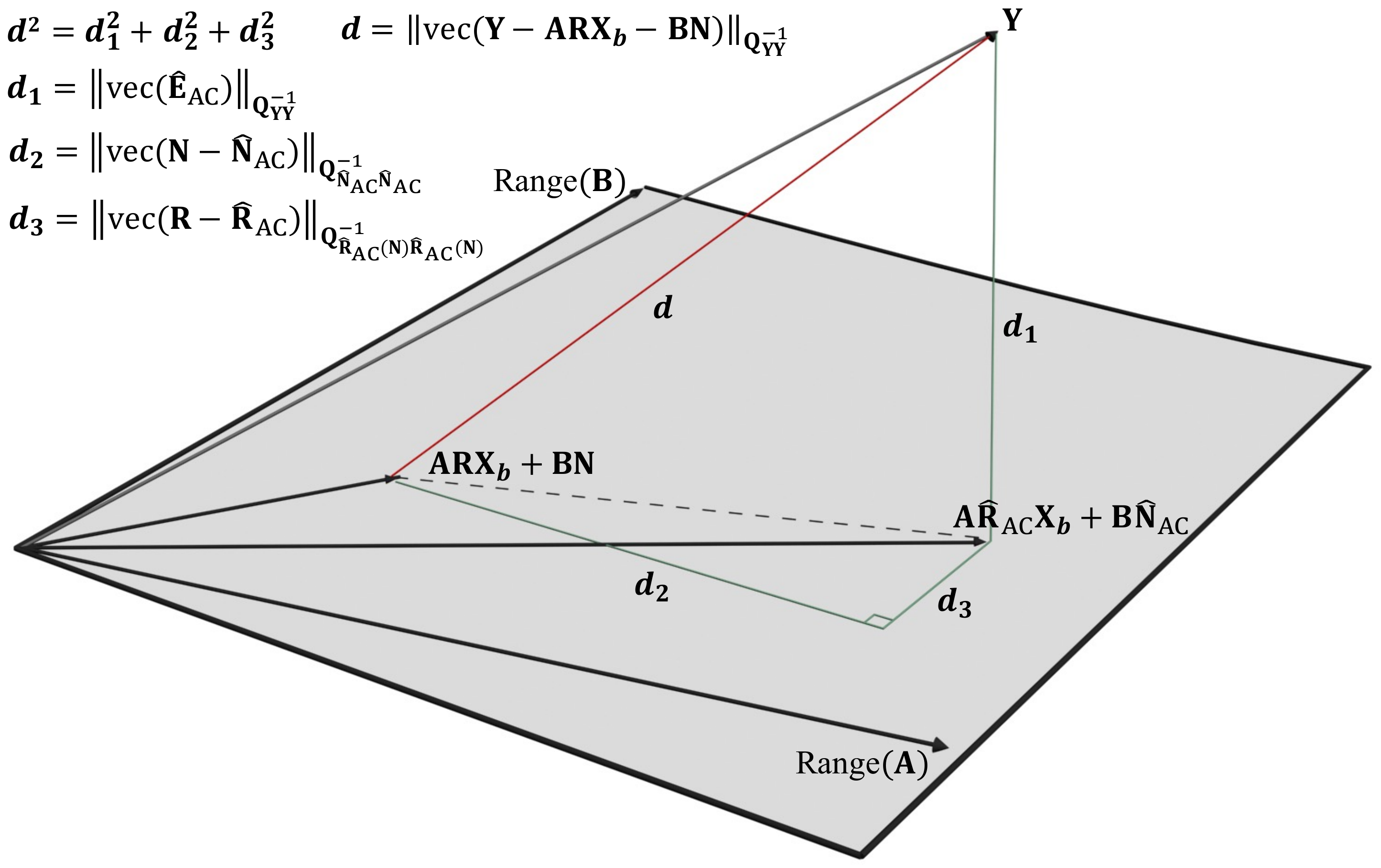}
} 
\subfigure[Proposed decomposition] { \label{fig:decomp1b} 
\includegraphics[width=0.98\columnwidth]{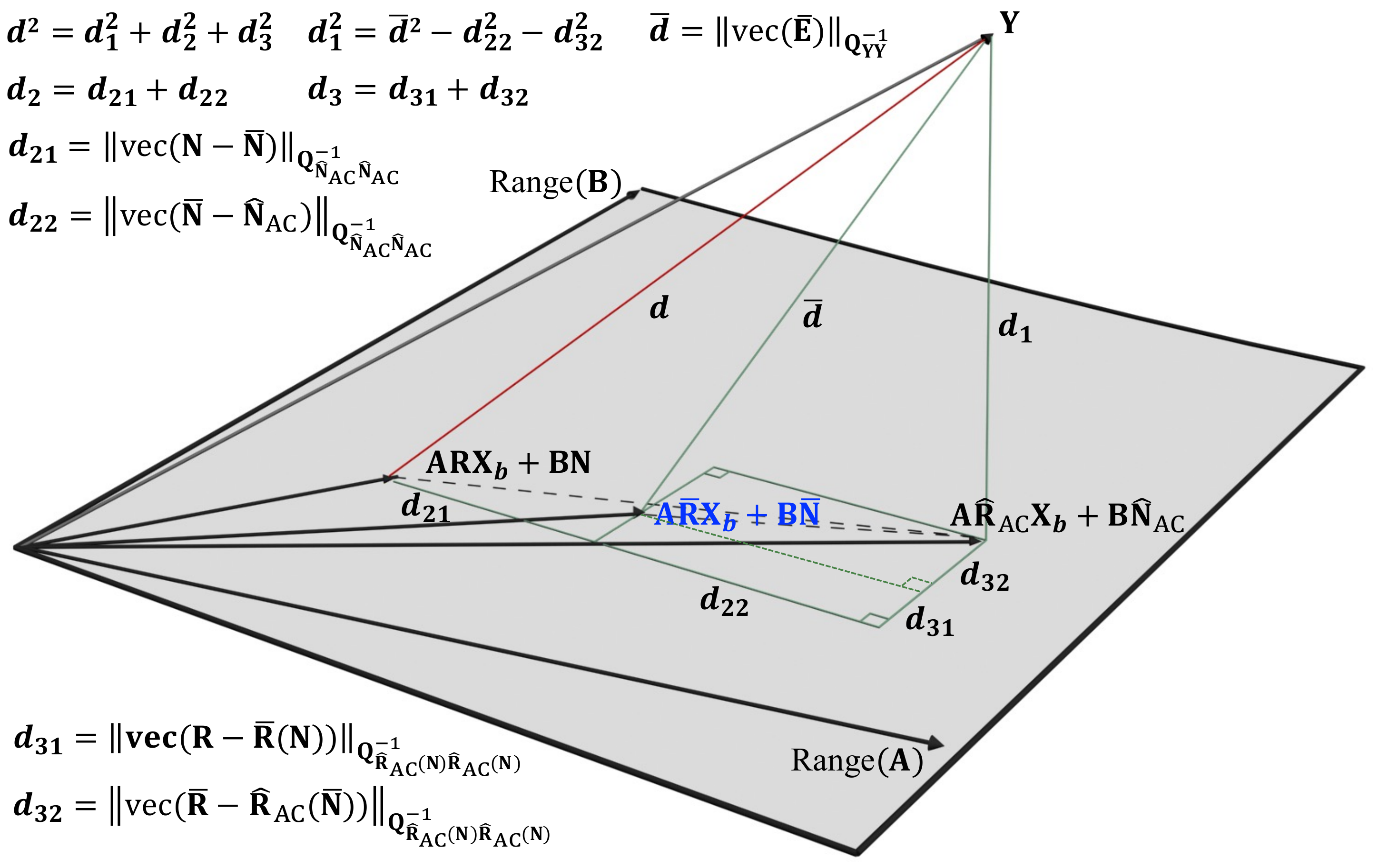} 
} 
\caption{Geometric illustration of the decomposition of the objective function.}
\label{fig:decomp}
\end{figure*}

\begin{lemma} \label{le:decom}
For $\forall {\mathbf{\bar R}} \in {\mathbb{R}^{3 \times q}}$ and $\forall \mathbf{\bar N} \in {\mathbb{R}^{\mathcal{S} \times \mathcal{A}}}$, the following relationship is an alternative form of \eref{eq:th1}
\begin{equation}
\resizebox{.88\hsize}{!}{
$
\begin{aligned}
&\left\| {\operatorname{vec}\!\left( \!{{\mathbf{Y}}\! -\!{\mathbf{AR}}{\mathbf{X}_b} \!- \!{\mathbf{BN}}}\! \right)}\! \right\|_{\mathbf{Q}_{\mathbf{Y\!Y}}^{-1}}^2 = \left\| {\operatorname{vec} \! \left( {\mathbf{\bar E}} \right)}\! \right\|_{\mathbf{Q}_{\mathbf{Y\!Y}}^{-1}}^2\\
& +\! \left\|\! \mathrm{vec} \!\left ( \!\mathbf{N} \!-\!\mathbf{\bar N}\!\right )\! \right\|_{\mathbf{Q}_{{\mathbf{\hat N}_{_\text{AC}}\!\mathbf{\hat N}_{_\text{AC}}}}^{-1}}^2\!\! \!+\! \left\| \mathrm{vec} \!\left ( \mathbf{R} \!-\!\mathbf{\hat R}_{_\text{AC}}\!\!\left (\mathbf{N} \right )\!\right )\! \right\|_{\mathbf{Q}_{{\mathbf{\hat R}_{_\text{AC}}}\!\left({\mathbf{N}}\right){\mathbf{\hat R}_{_\text{AC}}}\!\left({\mathbf{N}}\right)}^{-1}}^2\\
&+ 2\mathrm{vec}\! \left ( \mathbf{N} \!-\!\mathbf{\bar N}\right )^{\text{T}}{\mathbf{Q}_{{\mathbf{\hat N}_{_\text{AC}}\!\mathbf{\hat N}_{_\text{AC}}}}^{-1}}\!\mathrm{vec}\! \left ( \!\mathbf{\bar N} \!-\!\mathbf{\hat N}_{_\text{AC}}\!\right ) \\
 &- \left\| \!\mathrm{vec} \!\left ( \!\mathbf{\bar R} \!-\!\mathbf{\hat R}_{_\text{AC}}\!\!\left (\!\mathbf{\bar N} \!\right )\!\right ) \!\right\|^2_{\mathbf{Q}_{{\mathbf{\hat R}_{_\text{AC}}}\!\left({\mathbf{N}}\right){\mathbf{\hat R}_{_\text{AC}}}\!\left({\mathbf{N}}\right)}^{-1}}.
\end{aligned}
$}
\label{eq:le1}
\end{equation}
\end{lemma}
\begin{proof} 
See Appendix C.
\end{proof}

\lref{th:decom} and \lref{le:decom} hold for an arbitrary point, so as to ${{{\mathbf{\hat R}}}_{_\text{RM}}}$ and ${{{\mathbf{\hat N}}}_{_\text{RM}}}$. If we replace ${{\mathbf{\bar R}}}$ and ${{\mathbf{\bar N}}}$ using ${{{\mathbf{\hat R}}}_{_\text{RM}}}$ and ${{{\mathbf{\hat N}}}_{_\text{RM}}}$, the first and last terms on the right-hand side of \eref{eq:le1} are independent of the parameters. Then we achieve the equivalent form of \eref{eq:ocam} as follows
\begin{equation}
\mathop {\min }\limits_{ {\mathbf{N}} \in {\mathbb{Z}^{\mathcal{S} \times\mathcal{A}}}} \!\!\mathcal{C} \!\left ( \mathbf{N} \right ),
\label{eq:fdecom1}
\end{equation}
with
\begin{equation}
\resizebox{.88\hsize}{!}{
$
\begin{aligned}
\mathcal{C} \!\left ( \mathbf{N} \right )=&\left\|\! \mathrm{vec} \!\!\left ( \!\mathbf{N} \!-\!{{{\mathbf{\hat N}}}_{_\text{RM}}}\!\right )\! \right\|_{\mathbf{Q}_{{\mathbf{\hat N}_{_\text{AC}}\!\mathbf{\hat N}_{_\text{AC}}}}^{-1}}^2\!\! \!+\!2\mathrm{vec}\!\! \left ( \!\mathbf{N} \!-\!{{{\mathbf{\hat N}}}_{_\text{RM}}}\!\right )^{\text{T}}\!{\mathbf{Q}_{{\mathbf{\hat N}_{_\text{AC}}\!\mathbf{\hat N}_{_\text{AC}}}}^{-1}}\!\!\mathrm{vec}\!\! \left ( \!{{{\mathbf{\hat N}}}_{_\text{RM}}} \!-\!\mathbf{\hat N}_{_\text{AC}}\!\right )  \\
&  + \mathop {\min }\limits_{ {\mathbf{R}} \in {\mathbb{O}^{3 \times q}}} \!\left\| \mathrm{vec} \!\!\left ( \!\mathbf{R} \!-\!\mathbf{\hat R}_{_\text{AC}}\!\!\left (\!\mathbf{N} \!\right )\!\right )\! \right\|_{\mathbf{Q}_{{\mathbf{\hat R}_{_\text{AC}}}\!\left(\!{\mathbf{N}}\!\right){\mathbf{\hat R}_{_\text{AC}}}\!\left(\!{\mathbf{N}}\!\right)}^{-1}}^2.
\end{aligned}
$}
\label{eq:fdecom2}
\end{equation}
The function $\mathcal{C} \!\left ( \mathbf{N} \right )$ includes three coupled terms. 

If we ignore the last two terms in \eref{eq:fdecom2}, we can define a loose-form solution of RieMOCAD (RieMOCAD-LF); that is, we calculate the closest integer estimations to the float solution ${{{\mathbf{\hat N}}}_{_\text{RM}}}$ in the metric of the associated weight matrix 
\begin{equation}
\begin{aligned}
\mathop {\min }\limits_{ {\mathbf{N}} \in {\mathbb{Z}^{\mathcal{S} \times\mathcal{A}}}} \! \left\|\! \mathrm{vec} \!\!\left ( \!\mathbf{N} \!-\!{{{\mathbf{\hat N}}}_{_\text{RM}}}\!\right )\! \right\|_{\mathbf{Q}_{{\mathbf{\hat N}_{_\text{AC}}\!\mathbf{\hat N}_{_\text{AC}}}}^{-1}}^2.
\end{aligned}
\label{eq:simt2}
\end{equation}
RieMOCAD-LF, \eref{eq:simt2}, only partially employs the orthonormality constraint, which can be used as a low-complexity option in many practical applications. To fully incorporate this constraint, one needs to achieve the tight-form solution of RieMOCAD (RieMOCAD-TF), i.e., minimizing \eref{eq:fdecom1}. \sref{sec:search} presents the details to solve the optimization problem \eref{eq:fdecom1}. Note that \eref{eq:fdecom1} provides only one possible tight-form solution of the OC-AM, and \eref{eq:lambda3} is another alternative. In this contribution, \eref{eq:fdecom1} is applied rather than \eref{eq:lambda3} to accelerate the process of the integer ambiguity resolution.

\subsubsection{Integer Search Strategy}
\label{sec:search}
To seek for the optimal integer matrix that minimizes \eref{eq:fdecom1}, a search process is required and the search space is defined as
\begin{equation}
  \Omega \! \left ( \chi  \right )  =\left \{ {\mathbf{N}} \in {\mathbb{Z}^{\mathcal{S} \times\mathcal{A}}}  \left|\mathcal{C} \! \left ( {\mathbf{N}}  \right )\right. < \chi\right \}.
\end{equation}
It requires high computational load to evaluate the integer matrix in $\Omega \! \left ( \chi  \right )$ due to the complex optimization caused by the presence of the orthonormality constraint. In other words, we need to evaluate every candidate in $\Omega \! \left ( \chi  \right )$ by solving the orthonormality-constrained optimization using the methods discussed in \sref{sec:RieOpt}. Similar to the MC-LAMBDA method, the cost function $\mathcal{C} \!\left ( \mathbf{N} \right )$ is also bounded by two easier-to-evaluate functions, which are formulated as
\begin{equation}
\resizebox{.88\hsize}{!}{
$
\begin{aligned}
\mathcal{C}_{_{\text{L}}} \!\!\left ( \mathbf{N} \right )=&\left\|\! \mathrm{vec} \!\!\left ( \!\mathbf{N} \!-\!{{{\mathbf{\hat N}}}_{_\text{RM}}}\!\right )\! \right\|_{\mathbf{Q}_{{\mathbf{\hat N}_{_\text{AC}}\!\mathbf{\hat N}_{_\text{AC}}}}^{-1}}^2\!\! \!+\!2\mathrm{vec}\!\! \left ( \!\mathbf{N} \!-\!{{{\mathbf{\hat N}}}_{_\text{RM}}}\!\right )^{\text{T}}\!{\mathbf{Q}_{{\mathbf{\hat N}_{_\text{AC}}\!\mathbf{\hat N}_{_\text{AC}}}}^{-1}}\!\!\mathrm{vec}\!\! \left ( \!{{{\mathbf{\hat N}}}_{_\text{RM}}} \!-\!\mathbf{\hat N}_{_\text{AC}}\!\right )  \\
& + \xi_{\min} \! \sum_{i=1}^{q}\!\left( \left\|{\mathbf{\hat r}_{i}}\!\left(\!{\mathbf{N}}\!\right) \!\right\|-1\right)^2,\\
\mathcal{C}_{_{\text{U}}} \!\!\left ( \mathbf{N} \right )=&\left\|\! \mathrm{vec} \!\!\left ( \!\mathbf{N} \!-\!{{{\mathbf{\hat N}}}_{_\text{RM}}}\!\right )\! \right\|_{\mathbf{Q}_{{\mathbf{\hat N}_{_\text{AC}}\!\mathbf{\hat N}_{_\text{AC}}}}^{-1}}^2\!\! \!+\!2\mathrm{vec}\!\! \left ( \!\mathbf{N} \!-\!{{{\mathbf{\hat N}}}_{_\text{RM}}}\!\right )^{\text{T}}\!{\mathbf{Q}_{{\mathbf{\hat N}_{_\text{AC}}\!\mathbf{\hat N}_{_\text{AC}}}}^{-1}}\!\!\mathrm{vec}\!\! \left ( \!{{{\mathbf{\hat N}}}_{_\text{RM}}} \!-\!\mathbf{\hat N}_{_\text{AC}}\!\right )  \\
& + \xi_{\max} \! \sum_{i=1}^{q}\!\left( \left\|{\mathbf{\hat r}_{i}}\!\left(\!{\mathbf{N}}\!\right) \!\right\|-1\right)^2,
\end{aligned}
$}
\label{eq:bound}
\end{equation}
with
\begin{equation}
\mathcal{C}_{_{\text{L}}} \!\!\left ( \mathbf{N} \right ) \leq \mathcal{C} \!\left ( \mathbf{N} \right ) \leq \mathcal{C}_{_{\text{U}}} \!\!\left ( \mathbf{N} \right )
\end{equation}
where ${\mathbf{\hat r}_{i}}\!\left(\!{\mathbf{N}}\!\right)$ is the $i$-th column of ${\mathbf{\hat R}_{_\text{AC}}}\!\left(\!{\mathbf{N}}\!\right)$, and $\xi_{\min}$ and $\xi_{\max}$ denote the smallest and largest eigenvalues of ${\mathbf{Q}_{{\mathbf{\hat R}_{_\text{AC}}}\!\left({\mathbf{N}}\right){\mathbf{\hat R}_{_\text{AC}}}\!\left({\mathbf{N}}\right)}^{-1}}$, respectively.

Based on the lower and upper bounds of $\mathcal{C} \!\left ( \mathbf{N} \right )$, we can define the corresponding search space as
\begin{equation}
\begin{aligned}
  \Omega_{_{\text{L}}}\!\!\left ( \chi  \right )  =\left \{ {\mathbf{N}} \in {\mathbb{Z}^{\mathcal{S} \times\mathcal{A}}}  \left|\mathcal{C}_{_{\text{L}}} \!\! \left ( {\mathbf{N}}  \right )\right. < \chi\right \},\\
  \Omega_{_{\text{U}}}\!\!\left ( \chi  \right )  =\left \{ {\mathbf{N}} \in {\mathbb{Z}^{\mathcal{S} \times\mathcal{A}}}  \left|\mathcal{C}_{_{\text{U}}} \!\! \left ( {\mathbf{N}}  \right )\right. < \chi\right \},
  \end{aligned}
  \label{eq:setbound}
\end{equation}
with
\begin{equation}
\Omega_{_{\text{U}}}\!\!\left ( \chi  \right ) \subseteq \Omega\!\left ( \chi  \right ) \subseteq \Omega_{_{\text{L}}}\!\!\left ( \chi  \right ).
\end{equation}
Then, two search approaches, namely the search-and-shrink and search-and-expand strategies, can be applied to adjust the size of the search space adaptively so as to reduce the computational complexity significantly.

The search-and-shrink strategy takes advantage of the upper bound $\mathcal{C}_{_{\text{U}}} \!\!\left ( \mathbf{N} \right )$ and the corresponding set $\Omega_{_{\text{U}}}\!\!\left ( \chi  \right )$. Initially, a large $\chi_0>0$ is chosen to initialize $\chi$ such that $\Omega_{_{\text{U}}}\!\!\left ( \chi  \right )$ and $\Omega \!\left ( \chi  \right )$ are nonempty. Find an integer matrix $\mathbf{N}_1 \in \Omega_{_{\text{U}}}\!\!\left ( \chi  \right )$ with $\mathcal{C}_{_{\text{U}}} \!\!\left ( \mathbf{N}_1 \right ) = \chi_1 < \chi$. Then, shrink the search space by replacing $\chi = \chi_1$ and continue to look for an integer matrix $\mathbf{N}_2 \in \Omega_{_{\text{U}}}\!\!\left ( \chi  \right )$ with $\mathcal{C}_{_{\text{U}}} \!\!\left ( \mathbf{N}_2 \right ) = \chi_2 < \chi$. Iterate the search procedures until there is only one integer matrix $\mathbf{N}_{k}$ in the search space $\Omega_{_{\text{U}}}\!\!\left ( \chi  \right )$ (assume k iterations and $\chi = \chi_k$). $\mathbf{N} = \mathbf{N}_{k}$ minimizes $\mathcal{C}_{_{\text{U}}} \!\!\left ( \mathbf{N} \right )$ and $\chi =  \mathcal{C}_{_{\text{U}}} \!\!\left ( \mathbf{N}_{k} \right )$. Note that $\mathbf{N} = \mathbf{N}_{k}$ does have to minimize $\mathcal{C}\!\left ( \mathbf{N} \right )$. Therefore, we need to evaluate $\Omega\!\left (\chi  \right )$ using Riemannian optimization algorithms.
Finally, we seek for $\mathbf{N} \in \Omega\!\left ( \chi  \right )$ that minimizes $\mathcal{C}\!\left ( \mathbf{N} \right )$ as the estimation.

In contrast, the search-and-expand strategy utilizes the lower bound $\mathcal{C}_{_{\text{L}}} \!\!\left ( \mathbf{N} \right )$ and the associated search space $\Omega_{_{\text{L}}}\!\!\left ( \chi  \right )$. Initially, a small $\chi_0>0$ is chosen to initialize $\chi$. If $\Omega_{_{\text{L}}}\!\!\left ( \chi  \right )$ is empty, replace $\chi$ using a larger value $\chi_1$ and check $\Omega_{_{\text{L}}}\!\!\left ( \chi  \right )$ again. Iterate the search procedures (assume k steps) until $\Omega_{_{\text{L}}}\!\!\left ( \chi \right )$ with $\chi = \chi_k$ is nonempty. Then, evaluate all the candidates in $\Omega_{_{\text{L}}}\!\!\left ( \chi \right )$ using $\mathcal{C} \! \left ( {\mathbf{N}}  \right )$, that is
\begin{equation}
\Omega\!\left ( \chi  \right ) = \left \{ {\mathbf{N}} \in \Omega_{_{\text{L}}}\!\!\left ( \chi \right )  \left|\mathcal{C} \! \left ( {\mathbf{N}}  \right )\right. < \chi\right \}.
\label{eq:lowset}
\end{equation}
If $\Omega\!\left ( \chi  \right )$ is nonempty, we extract $\mathbf{N} \in \Omega\!\left ( \chi  \right )$ that minimizes $\mathcal{C}\!\left ( \mathbf{N} \right )$ as the finial estimation. Otherwise, let $\chi = \chi_{k+1} > \chi_{k}$, $k = k+1$ and repeat the search procedures until $\Omega\!\left ( \chi \right )$ is nonempty. 



\begin{algorithm}[t]
\begin{algorithmic}[1]
\STATE Compute ${{{\mathbf{\hat R}}}_{_\text{AC}} }$ and ${{{\mathbf{\hat N}}}_{_\text{AC}} }$ based on \eref{eq:float2}.
\STATE Solve \eref{eq:rm1} to obtain ${{{\mathbf{\hat R}}}_{_\text{RM}}}$ using the Riemannian algorithms.
\STATE Compute ${{{\mathbf{\hat N}}}_{_\text{RM}}}$ using \eref{eq:rm1cond1} and \eref{eq:rm1cond}.
\STATE Initialize $k=0$ and $\chi = \chi_0>0$ such that $\Omega_{_{\text{U}}}\!\!\left ( \chi  \right )\neq \varnothing$ and $\Omega\!\left ( \chi  \right ) \neq \varnothing$.
\WHILE {$\text{card}\{\Omega_{_{\text{U}}}\!\!\left ( \chi  \right )\}>1$}
\STATE $k = k+1$.
\STATE Find $\mathbf{N}_k \in \Omega_{_{\text{U}}}\!\!\left ( \chi  \right )$ with $\mathcal{C}_{_{\text{U}}} \!\!\left ( \mathbf{N}_k \right ) = \chi_k < \chi$.
\STATE $\chi=\chi_k$.
\STATE Update $\Omega_{_{\text{U}}}\!\!\left ( \chi  \right )$ based on \eref{eq:bound} and \eref{eq:setbound}.
\ENDWHILE
\STATE Establish $\Omega\!\left ( \chi  \right )$ using the Riemannian algorithms.
\STATE Choose ${\mathbf{N}} \!\in\! \Omega\!\left ( \chi  \right )$ that minimizes $\mathcal{C}\!\left ( \mathbf{N} \right )$.
\end{algorithmic}
\caption{RieMOCAD using the search-and-shrink strategy.}
\label{alg0_shr}
\end{algorithm}

\begin{algorithm}[t]
\begin{algorithmic}[1]
\STATE Compute ${{{\mathbf{\hat R}}}_{_\text{AC}} }$ and ${{{\mathbf{\hat N}}}_{_\text{AC}} }$ based on \eref{eq:float2}.
\STATE Solve \eref{eq:rm1} to obtain ${{{\mathbf{\hat R}}}_{_\text{RM}}}$ using the Riemannian algorithms.
\STATE Compute ${{{\mathbf{\hat N}}}_{_\text{RM}}}$ using \eref{eq:rm1cond1} and \eref{eq:rm1cond}.
\STATE Initialize $\chi = \chi_0>0$ and $k = 0$.
\STATE Let $\Omega\!\left ( \chi  \right ) = \varnothing$.
\WHILE {$\Omega\!\left ( \chi  \right )=\varnothing$}
\STATE Set $\Omega_{_{\text{L}}}\!\!\left ( \chi  \right )$ based on \eref{eq:bound} and \eref{eq:setbound}.
\FORALL {${\mathbf{N}} \!\in\! \Omega_{_{\text{L}}}\!\!\left ( \chi  \right )$}
\STATE Calculate $\mathcal{C} \!\left ( \mathbf{N} \right )$ using the Riemannian algorithms.

\IF {$\mathcal{C} \!\left ( \mathbf{N} \right )<\chi$}
\STATE $\Omega\!\left ( \chi  \right ) = \Omega\!\left ( \chi  \right ) \cup {\mathbf{N}}$.
\ENDIF
\ENDFOR
\IF {$\Omega\!\left ( \chi  \right )=\varnothing$}
\STATE $\chi = \chi_{k+1} >\chi_{k}$.
\STATE $k = k + 1$.
\ENDIF
\ENDWHILE
\STATE Choose ${\mathbf{N}} \!\in\! \Omega\!\left ( \chi  \right )$ that minimizes $\mathcal{C}\!\left ( \mathbf{N} \right )$.
\end{algorithmic}
\caption{RieMOCAD using the search-and-expand strategy.}
\label{alg0_exp}
\end{algorithm}

\algref{alg0_shr} and \algref{alg0_exp} summarize the RieMOCAD method (tight form) using the search-and-shrink and search-and-expand strategies, respectively. Compared with the existing methods, the major differences of the proposed approach are as follows. The RieMOCAD method requires the float solutions Least Squares and Riemannian optimization, whereas the MC-LAMBDA method requires only the Least Squares. The decomposition of the objective function is different between the RieMOCAD and MC-LAMBDA methods, leading to different cost functions used in the feasible set. Instead of solving the orthonormality-constrained least-squares problem through the gradient descent techniques, the RieMOCAD method employs the optimization over the Riemannian manifold in the search process.

\section{Performance Evaluation}
\label{sec: SimExp}
This section tests the proposed RieMOCAD method and compares it with several benchmarks using simulations as well as experimental data. For different scenarios, we analyze performance in terms of the quality of the float solution, success rate, and computational efficiency in the single-epoch and single-frequency case, with only GPS satellites utilized. Riemannian optimization algorithms are implemented based on the Manopt toolbox \cite{manopt}

\begin{figure*}[htbp]
\centering
\includegraphics[width=1.0\textwidth]{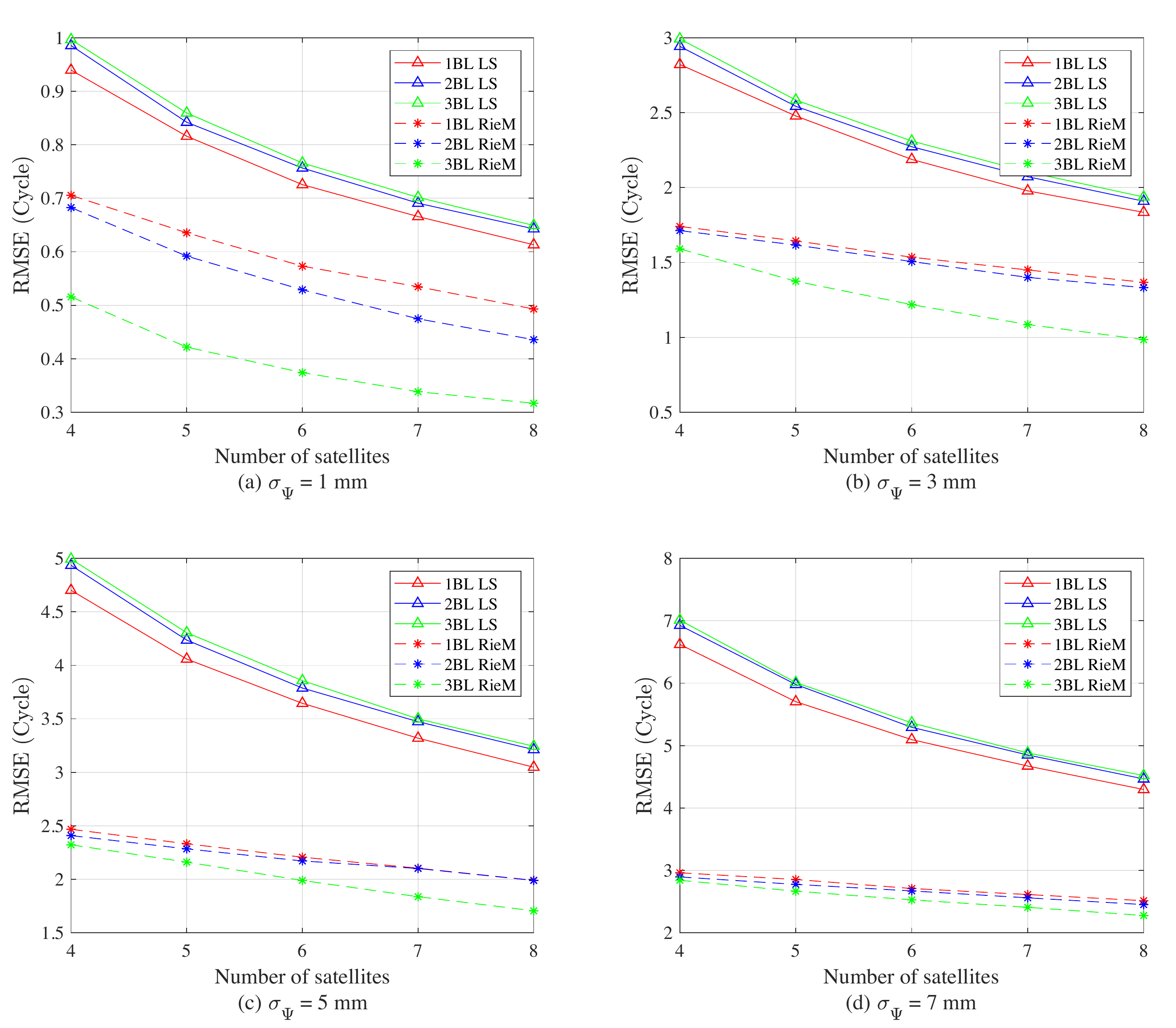}
\caption{RMSE of float ambiguity estimations (a single baseline to three baselines).}
\label{fig:AveFloatErr1}
\end{figure*}
\begin{figure*}[htbp]
\centering
\includegraphics[width=1.0\textwidth]{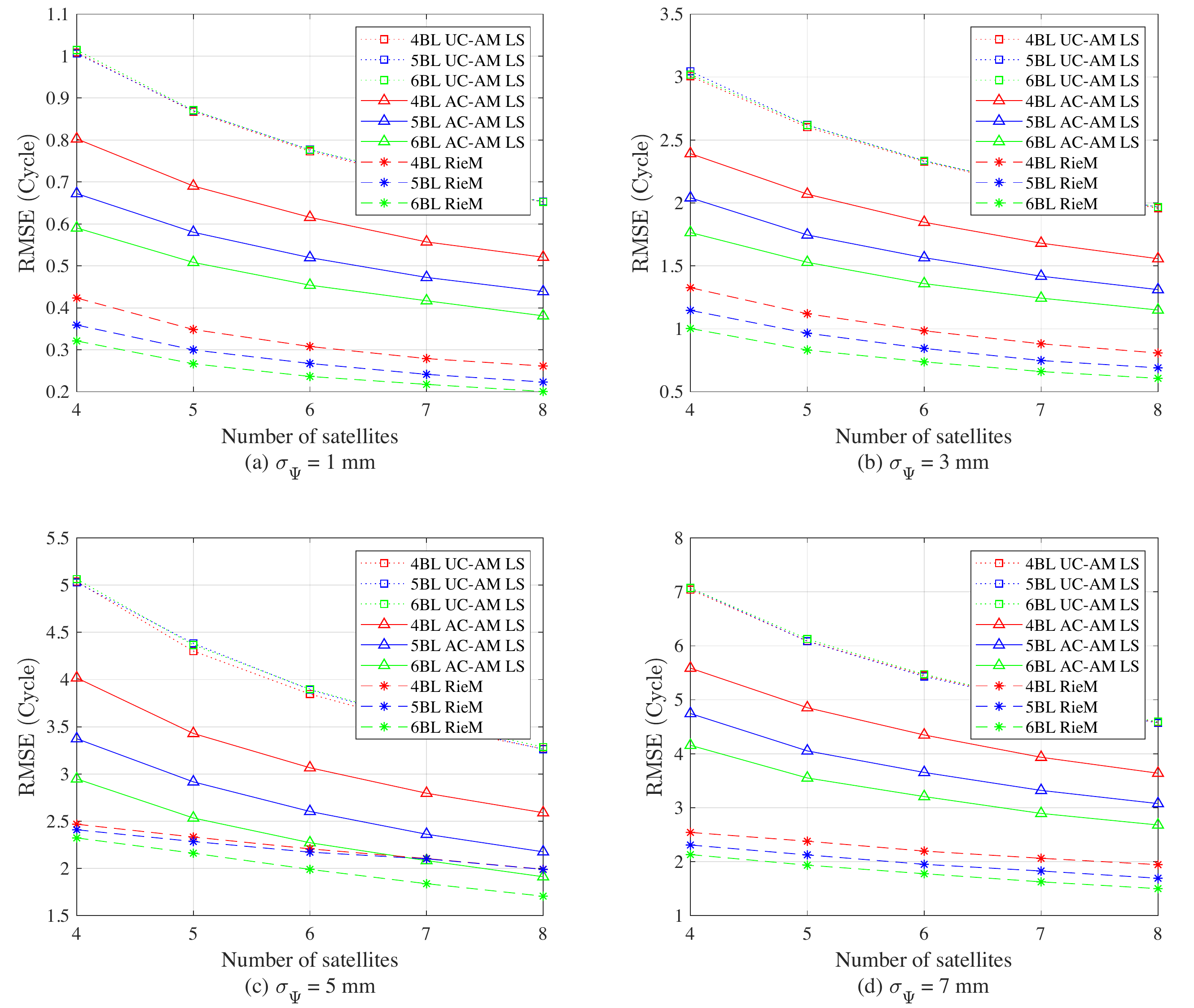}
\caption{RMSE of float ambiguity estimations (four baselines to six baselines).}
\label{fig:AveFloatErr2}
\end{figure*}

\begin{table*}[htbp]
\centering
  \caption{Success rate (\%) for different measurement precision $\sigma_{\Psi} (\text{mm})$, number of tracked satellites (\#Sat), and number of baselines. \\UC-AM/AC-AM  \\RieMOCAD-LF}
    \begin{tabular}{c|cccc|cccc|cccc}
    \hline
     & \multicolumn{4}{c|}{Single Baseline} & \multicolumn{4}{c|}{Two Baselines} & \multicolumn{4}{c}{Three Baselines}\\
    \hline
    \diagbox{\#Sat}{$\sigma_{\Psi}$}  &7 &5 &3 &1       &7 &5 &3 &1  &7 &5 &3 &1\\
    \hline
    \multirow{2}[2]{*}{4}

           &0.01  &0.06  &0.30  &6.84     &0  &0  &0  &0.75        &0  &0  &0  &0.11  \\
           &\textbf{2.74}  &\textbf{3.76}  &\textbf{6.00}  &\textbf{22.53}     &\textbf{0.35}  &\textbf{0.45}  &\textbf{1.14}  &\textbf{11.63}      &\textbf{0.06}  &\textbf{0.16}  &\textbf{0.92}  &\textbf{18.07}  \\

    \hline
    \multirow{2}[2]{*}{5}

           &0.19  &0.54  &3.29  &61.60     &0  &0  &0.18  &48.85         &0  &0  &0.02  &40.66  \\
           &\textbf{2.94}  &\textbf{5.20}  &\textbf{13.26}  &\textbf{73.74}     &\textbf{0.22}  &\textbf{0.75}  &\textbf{3.84}  &\textbf{68.75}      &\textbf{0.19}  &\textbf{0.20}  &\textbf{3.11}  &\textbf{76.01}  \\       
    \hline
    \multirow{2}[2]{*}{6} 

           &0.83  &3.90  &25.07  &99.11     &0  &0.24  &9.83  &99.19             &0  &0.02  &4.45  &99.50  \\
           &\textbf{5.73}  &\textbf{13.75}  &\textbf{43.48}  &\textbf{99.35}     &\textbf{0.67}  &\textbf{3.95}  &\textbf{29.25}  &\textbf{99.54}        &\textbf{0.26}  &\textbf{2.38}  &\textbf{29.92}  &\textbf{99.96}  \\   
    \hline
    \multirow{2}[2]{*}{7}  

           &4.13  &16.70  &68.34  &\textbf{99.98}     &0.38  &5.38  &60.07  &\textbf{99.99}           &0.04  &1.74  &56.75  &\textbf{100}  \\
           &\textbf{13.47}  &\textbf{33.43}  &\textbf{79.78}  &\textbf{99.98}     &\textbf{3.76}  &\textbf{18.67}  &\textbf{76.64}  &\textbf{99.99}          &\textbf{1.85}  &\textbf{14.77}  &\textbf{81.03}  &\textbf{100} \\
    \hline
    \multirow{2}[2]{*}{8}  

           &13.41  &46.71  &94.62  &\textbf{100}     &3.90  &33.86  &94.51  &\textbf{100}         &1.35  &26.85  &94.84  &\textbf{100}  \\
           &\textbf{29.06}  &\textbf{63.78}  &\textbf{96.75}  &\textbf{100}     &\textbf{16.29}  &\textbf{57.14}  &\textbf{97.14}  &\textbf{100}         &\textbf{11.60}  &\textbf{56.91}  &\textbf{98.03}  &\textbf{100}  \\
    \hline
    \end{tabular}%
  \label{tab:suc1}%
\end{table*}%

\begin{table*}[htbp]
\centering
  \caption{Success rate (\%) for different measurement precision $\sigma_{\Psi} (\text{mm})$, number of tracked satellites (\#Sat), and number of baselines. \\UC-AM  \\AC-AM  \\RieMOCAD-LF}
    \begin{tabular}{c|cccc|cccc|cccc}
    \hline
     & \multicolumn{4}{c|}{Four Baselines} & \multicolumn{4}{c|}{Five Baselines} & \multicolumn{4}{c}{Six Baselines}\\
    \hline
    \diagbox{\#Sat}{$\sigma_{\Psi}$}  &7 &5 &3 &1       &7 &5 &3 &1  &7 &5 &3 &1\\
    \hline
    \multirow{3}[2]{*}{4}

           &0  &0  &0  &0.01     &0  &0  &0  &0.01        &0  &0  &0  &0  \\
           &0  &0  &0.03  &40.78     &0  &0.02  &4.10  &96.48      &0.03  &1.68  &54.60  &\textbf{99.89}  \\
           &\textbf{0.03}  &\textbf{0.17}  &\textbf{2.82}  &\textbf{75.91}     &\textbf{0.13}  &\textbf{1.64}  &\textbf{23.67}  &\textbf{98.41}         &\textbf{1.38}  &\textbf{13.40}  &\textbf{73.47}  &\textbf{99.89}  \\

    \hline
    \multirow{3}[2]{*}{5}

           &0  &0  &0  &35.26     &0  &0  &0  &30.96         &0  &0  &0  &27.91  \\
           &0  &0.13  &13.33  &96.44     &0.56  &14.42  &81.12  &99.97      &15.34  &74.92  &97.59  &\textbf{100}  \\      
           &\textbf{0.33}  &\textbf{3.54}  &\textbf{38.26}  &\textbf{98.72}     &\textbf{6.02}  &\textbf{37.51}  &\textbf{88.72}  &\textbf{99.99}         &\textbf{37.71}  &\textbf{83.80}  &\textbf{98.71}  &\textbf{100}  \\
    \hline
    \multirow{3}[2]{*}{6} 

           &0  &0.02  &1.98  &99.47     &0  &0  &1.19  &99.59             &0  &0  &0.73  &99.62  \\
           &0.72  &16.78  &88.09  &\textbf{100}     &31.20  &91.65  &99.70  &\textbf{100}         &91.70  &99.54  &\textbf{99.99}  &\textbf{100}  \\   
           &\textbf{7.22}  &\textbf{42.51}  &\textbf{93.96}  &\textbf{100}     &\textbf{57.96}  &\textbf{95.46}  &\textbf{99.82}  &\textbf{100}         &\textbf{95.28}  &\textbf{99.63}  &\textbf{99.99}  &\textbf{100}  \\
    \hline
    \multirow{3}[2]{*}{7}  

           &0  &0.78  &51.69  &\textbf{100}     &0  &0.38  &49.41  &\textbf{100}           &0  &0.14  &46.86  &\textbf{100}  \\
           &17.78  &80.22  &99.12  &\textbf{100}     &91.81  &99.28  &\textbf{100}  &\textbf{100}          &99.72  &\textbf{99.96}  &\textbf{100}  &\textbf{100}  \\
           &\textbf{43.27}  &\textbf{89.35} &\textbf{99.69}  &\textbf{100}     &\textbf{96.08}  &\textbf{99.61}  &\textbf{100}  &\textbf{100}         &\textbf{99.79}  &\textbf{99.96}  &\textbf{100}  &\textbf{100}  \\
    \hline
    \multirow{3}[2]{*}{8}  

           &0.56  &23.75  &94.35  &\textbf{100}     &0.29  &20.29  &94.32  &\textbf{100}          &0.09  &17.24  &94.11  &\textbf{100}  \\
           &73.20  &97.80  &99.99  &\textbf{100}     &99.55  &\textbf{99.99}  &\textbf{100}  &\textbf{100}         &99.98  &\textbf{100}  &\textbf{100}  &\textbf{100}  \\
           &\textbf{86.13}  &\textbf{98.84}  &\textbf{100}  &\textbf{100}     &\textbf{99.69}  &\textbf{99.99}  &\textbf{100}  &\textbf{100}         &\textbf{99.99}  &\textbf{100}  &\textbf{100}  &\textbf{100}  \\
    \hline
    \end{tabular}%
  \label{tab:suc2}%
\end{table*}%

\begin{table*}[htbp]
\centering
  \caption{Success rate (\%) for different measurement precision $\sigma_{\Psi} (\text{mm})$, number of tracked satellites (\#Sat), and number of baselines. \\MC-LAMBDA  \\RieMOCAD-TF}
    \begin{tabular}{c|cccc|cccc|cccc}
    \hline
     & \multicolumn{4}{c|}{Single Baseline} & \multicolumn{4}{c|}{Two Baselines} & \multicolumn{4}{c}{Three Baselines}\\
    \hline
    \diagbox{\#Sat}{$\sigma_{\Psi}$}  &7 &5 &3 &1       &7 &5 &3 &1  &7 &5 &3 &1\\
    \hline
    \multirow{2}[2]{*}{4}

           &8.39  &14.46  &28.00  &\textbf{81.18}     &0  &0  &4.45  &97.91        &0  &0  &0.12  &50.43  \\
           &\textbf{10.19}  &\textbf{14.82}  &\textbf{28.01}  &\textbf{81.18}     &\textbf{6.39}  &\textbf{12.88}  &\textbf{39.32}  &\textbf{98.39}       &\textbf{2.25}  &\textbf{6.88}  &\textbf{16.55}  &\textbf{92.22}  \\

    \hline
    \multirow{2}[2]{*}{5}

           &21.80  &37.74  &70.79  &\textbf{99.20}     &1.76  &11.78  &79.99  &\textbf{100}        &0  &0.29  &20.11  &\textbf{100}  \\
           &\textbf{21.91}  &\textbf{37.77}  &\textbf{70.81}  &99.19     &\textbf{20.90}  &\textbf{44.89}  &\textbf{92.21}  &\textbf{100}        &\textbf{3.99}  &\textbf{15.38}  &\textbf{67.75}  &\textbf{100}  \\       
    \hline
    \multirow{2}[2]{*}{6} 

           &46.54  &\textbf{71.67} &95.41  &\textbf{99.98}     &18.44  &71.96  &\textbf{99.85}  &\textbf{100}        &0.88  &16.79  &97.36  &\textbf{100}  \\
           &\textbf{46.55}  &\textbf{71.67}  &\textbf{95.42}  &\textbf{99.98}     &\textbf{55.56}  &\textbf{91.77}  &\textbf{99.85}  &\textbf{100}        &\textbf{17.99}  &\textbf{61.80}  &\textbf{99.09}  &\textbf{100}  \\   
    \hline
    \multirow{2}[2]{*}{7}  

           &\textbf{72.80}   &\textbf{92.10}   &\textbf{99.42}   &\textbf{100}      &66.12  &98.72  &\textbf{100}   &\textbf{100}         &16.55  &84.63  &\textbf{100}   &\textbf{100}   \\
           &\textbf{72.80}   &\textbf{92.10}   &\textbf{99.42}   &\textbf{100}      &\textbf{89.55}   &\textbf{99.11}   &\textbf{100}   &\textbf{100}        &\textbf{56.31}   &\textbf{97.01}   &\textbf{100}   &\textbf{100}   \\
    \hline
    \multirow{2}[2]{*}{8}  

           &\textbf{88.19}  &\textbf{98.00}  &\textbf{99.92}  &\textbf{100}     &96.37  &\textbf{100}  &\textbf{100}  &\textbf{100}        &73.26  &\textbf{100}  &\textbf{100}  &\textbf{100}  \\
           &\textbf{88.19}  &\textbf{98.00}  &\textbf{99.92}  &\textbf{100}     &\textbf{99.23}  &\textbf{100}  &\textbf{100}  &\textbf{100}        &\textbf{93.11}  &\textbf{100}  &\textbf{100}  &\textbf{100}  \\
    \hline
    \end{tabular}%
  \label{tab:suc3}%
\end{table*}%

\begin{table*}[htbp]
\centering
  \caption{Number of integers in the search space for different measurement precision $\sigma_{\Psi} (\text{mm})$, number of tracked satellites (\#Sat), and number of baselines. \\MC-LAMBDA  \\RieMOCAD-TF}
  \scalebox{0.96}{
    \begin{tabular}{c|cccc|cccc|cccc}
    \hline
     & \multicolumn{4}{c|}{Single Baseline} & \multicolumn{4}{c|}{Two Baselines} & \multicolumn{4}{c}{Three Baselines}\\
    \hline
    \diagbox{\#Sat}{$\sigma_{\Psi}$}  &7 &5 &3 &1       &7 &5 &3 &1  &7 &5 &3 &1\\
    \hline
    \multirow{2}[2]{*}{4}
    &4268.66 &2164.82 &595.79 &39.13  &10000 &9992.12 &9990.35 &3097.07  &\textbf{10000} &\textbf{10000} &\textbf{10000} &9640.19\\
    &\textbf{420.28} &\textbf{277.91} &\textbf{140.19} &\textbf{23.37}  &\textbf{9960.78} &\textbf{9914.55} &\textbf{9776.32}  &\textbf{1607.71} &\textbf{10000} &\textbf{10000} &\textbf{10000} &\textbf{7163.28}\\  
    \hline
    \multirow{2}[2]{*}{5}
    &1403.85 &498.66 &97.43  &3.75 &9987.01 &9865.44  &6726.23   &12.09 &\textbf{10000} &10000 &9977.65 &52.43\\ 
    &\textbf{253.20} &\textbf{149.77} &\textbf{51.17} &\textbf{3.08} &\textbf{9873.52} &\textbf{9285.21}  &\textbf{4614.93}  &\textbf{7.22} &\textbf{10000} &\textbf{9998.11}  &\textbf{9314.99} &\textbf{30.11}\\        
    \hline
    \multirow{2}[2]{*}{6} 
    &496.25 &137.22  &17.29 &2.03  &9761.09 &7260.33 &415.42 &2.04 &10000 &9915.11 &3838.67 &2.07\\
    &\textbf{153.19} &\textbf{64.16} &\textbf{10.98} &\textbf{2.00} &\textbf{8892.29} &\textbf{4888.82} &\textbf{187.58}  &\textbf{2.00} &\textbf{9971.03} &\textbf{9098.00} &\textbf{1391.81} &\textbf{2.01}  \\      
    \hline
    \multirow{2}[2]{*}{7}  
    &177.78 &32.46 &3.53 &\textbf{2.00} &7238.69 &1414.66 &10.01 &\textbf{2.00} &9914.09 &6567.13 &34.27 &\textbf{2.00}\\
    &\textbf{74.23} &\textbf{18.09} &\textbf{2.85} &\textbf{2.00} &\textbf{5321.99} &\textbf{692.21}  &\textbf{5.90} &\textbf{2.00} &\textbf{9158.31} &\textbf{3335.94} &\textbf{12.12} &\textbf{2.00} \\   
    \hline
    \multirow{2}[2]{*}{8}  
    &58.86  &7.91  &2.20  &\textbf{2.00} &2679.34 &82.03 &2.31 &\textbf{2.00} &7650.06 &469.38 &2.74 &\textbf{2.00}\\
    &\textbf{30.22} &\textbf{5.60} &\textbf{2.01} &\textbf{2.00} &\textbf{1338.81} &\textbf{32.24}   &\textbf{2.07}  &\textbf{2.00} &\textbf{4607.26} &\textbf{116.35} &\textbf{2.45} &\textbf{2.00}  \\   
    \hline
    \end{tabular}%
    }
  \label{tab:num1}%
\end{table*}%

\subsection{Simulation Analysis}
Simulated data is used to test different methods under error-controlled conditions. The simulations are implemented using the assumed platform's attitude and the actual satellite orbit information in the GPS Yuma Almanacs file on November 7, 2021. Several scenarios are simulated with a different number of tracked satellites, a different number of antennas, and various measurement noises. The number of satellites is varied from four to eight, i.e., from the satellite-deprived situations to satellite-sufficient environments, by randomly choosing the satellites from the visible ones. The noise level is controlled by adding a zero-mean Gaussian noise with a specified standard deviation to the undifferenced GNSS observations. We assume that the noise of pseudo-range data is two orders of magnitude higher than that of the carrier phase, that is, $\sigma_{_\mathbf{P}} = 100 \sigma_{_\Psi} $. For each simulated scenario, a set of $10^4$ data is generated, with a random rotation and random antenna positions in the body frame in each trial. We let the baseline length be 1 meter and guarantee the antenna position matrix to be full row rank.

Firstly, we will study the quality of various float ambiguity solutions, i.e., ${\mathbf{\hat N}_{_\text{UC}}}$, ${\mathbf{\hat N}_{_\text{AC}}}$ and ${\mathbf{\hat N}_{_\text{RM}}}$ obtained by \eref{eq:float1}, \eref{eq:float2}, and \eref{eq:floatrm}, respectively. We characterize the quality of the float solution in terms of two metrics. The first metric is the root-mean-square error (RMSE) that captures the differences between the float ambiguity solutions and the ground truth integer ambiguities. The second metric is defined as the success rate of the optimal integers of \eref{eq:lambda1}, \eref{eq:lambda2}, and \eref{eq:simt2}, which gives us an idea about how good the closest integer estimations to different float solutions in the metric of the associated weight matrix.

\fref{fig:AveFloatErr1} and \fref{fig:AveFloatErr2} demonstrate the RMSE of the float ambiguities in various simulated environments. When the number of baselines (BL) is less or equal to the dimension of the range of the baseline matrix, the UC-AM and AC-AM have exactly the same float solution, that is, ${\mathbf{\hat N}_{_\text{UC}}} = {\mathbf{\hat N}_{_\text{AC}}}$, resulting in the same fixed integer solutions of \eref{eq:lambda1} and \eref{eq:lambda2}. Hence, \fref{fig:AveFloatErr1} shows only the LS solutions and Riemannian-manifold-based (RieM) solutions.  According to \fref{fig:AveFloatErr1}, the more baselines we use, the larger the error of the LS estimation. On the contrary, the RieM-based errors decrease as the number of baselines increases. The reason why multi-baseline configuration improves the RieM-based float solution is that the set-up benefits from more antenna-array geometry information.

As shown in \fref{fig:AveFloatErr2}, when the number of baselines is greater than the dimension of the range of the baseline matrix, the quality of the float solutions for the UC-AM tends to converge to a stable level. However, the AC-AM demonstrates its advantage compared with the UC-AM since the affine transformation (linear constraint) can take advantage of the redundancy of the GNSS data. To elaborate, we can analyze each model's number of observations and unknowns. For the UC-AM, there are $2{ \mathcal{A} \mathcal{S} }$ observations and ${ \mathcal{A} \mathcal{S} } + 3{ \mathcal{A}}$ unknown parameters. The AC-AM involves the same measurements as the UC-AM, with only ${ \mathcal{A} \mathcal{S} } + 9$ unknowns ($9< 3{ \mathcal{A}}$). The RieMOCAD-LF has the same number of observables and unknowns as the AC-AM. However, the additional nonlinear constraint explains why the RieMOCAD-LF outperforms the other two benchmarks. Unlike the UC-AM, the performance of the AC-AM and RieMOCAD-LF improves as the number of baselines rises. 

\tref{tab:suc1} and \tref{tab:suc2} list the success rate of the UC-AM, AC-AM and RieMOCAD-LF, that is, the optimums of \eref{eq:lambda1}, \eref{eq:lambda2}, and \eref{eq:simt2}. The success rate of the UC-AM continuously decreases as the number of baselines increases caused by the low-quality float solution coupled with the growth of dimension of the integer space. Generally, the AC-AM and RieMOCAD-LF tend to achieve a higher success rate when more baselines are considered since the improved float solutions are offered. However, the performance in some setups may not follow the trend. It depends on which factor, the quality of the float solution, or the dimension of the integer space, dominates the results. As shown in \tref{tab:suc2}, the AC-AM and RieMOCAD-LF can offer acceptable results in many scenarios (except for the extreme ones, where the tight-form solutions of the OC-AM are required) when using a large number of baselines. \fref{fig:AveFloatErr1}, \fref{fig:AveFloatErr2}, \tref{tab:suc1} and \tref{tab:suc2} indicate the significant improvement of the RieM-based float solution compared with the LS solution. 


According to \tref{tab:suc1} and \tref{tab:suc2}, the estimations of the UC-AM, AC-AM, and RieMOCAD-LF are not up to par in many challenging environments, especially when the number of baselines is less than the dimension of the range space of baseline matrix. Therefore, the robust model OC-AM is required to improve the results. Both \eref{eq:lambda3} and \eref{eq:fdecom1} are just two alternative forms of the optimization \eref{eq:ocam}. Therefore, if the search space is large enough, \eref{eq:ocam}, \eref{eq:lambda3} and \eref{eq:fdecom1} should achieve identical optima.

This work aims to speed up the search using the proposed intermediate solution to shrink the search space required. To avoid an enormous search space caused by the adaptive adjustment procedures discussed in \sref{sec:search}, we set up $10^4$ integer candidates as the limit to evaluate. \tref{tab:suc3} demonstrates the success rate for the MC-LAMBDA method and RieMOCAD-TF, i.e., \eref{eq:lambda3} and \eref{eq:fdecom1}, under different simulated environments. \tref{tab:num1} summarizes the corresponding number of integer candidates that need to evaluate. In a single-baseline configuration, RieMOCAD-TF and the MC-LAMBDA method have a comparable success rate, except for the extremely challenging set-up with $\sigma_{_\Psi} = 7 \text{ mm}$, $\sigma_{_\mathbf{P}} = 70 \text{ cm}$ and 4 satellites, where RieMOCAD-TF offers a better success-rate performance. However, the required search space of RieMOCAD-TF is smaller than that of the MC-LAMBDA method, indicating an improvement in computational efficiency. In dual- or triple-baseline cases, it is more difficult to resolve the integer ambiguities due to the high dimensionality of the integer space. In such situations, the advantages of the proposed approach become more notable. As stated above, the different success rate derives from the limit of the size of the search space, which also verifies the low complexity of RieMOCAD-TF. In summary, integrating Riemannian algorithms into the process of ambiguity resolution can shrink the search space in the integer domain at the expense of just an additional nonlinear optimization.


\begin{figure}[htbp]
\centering
\includegraphics[width=0.45\textwidth]{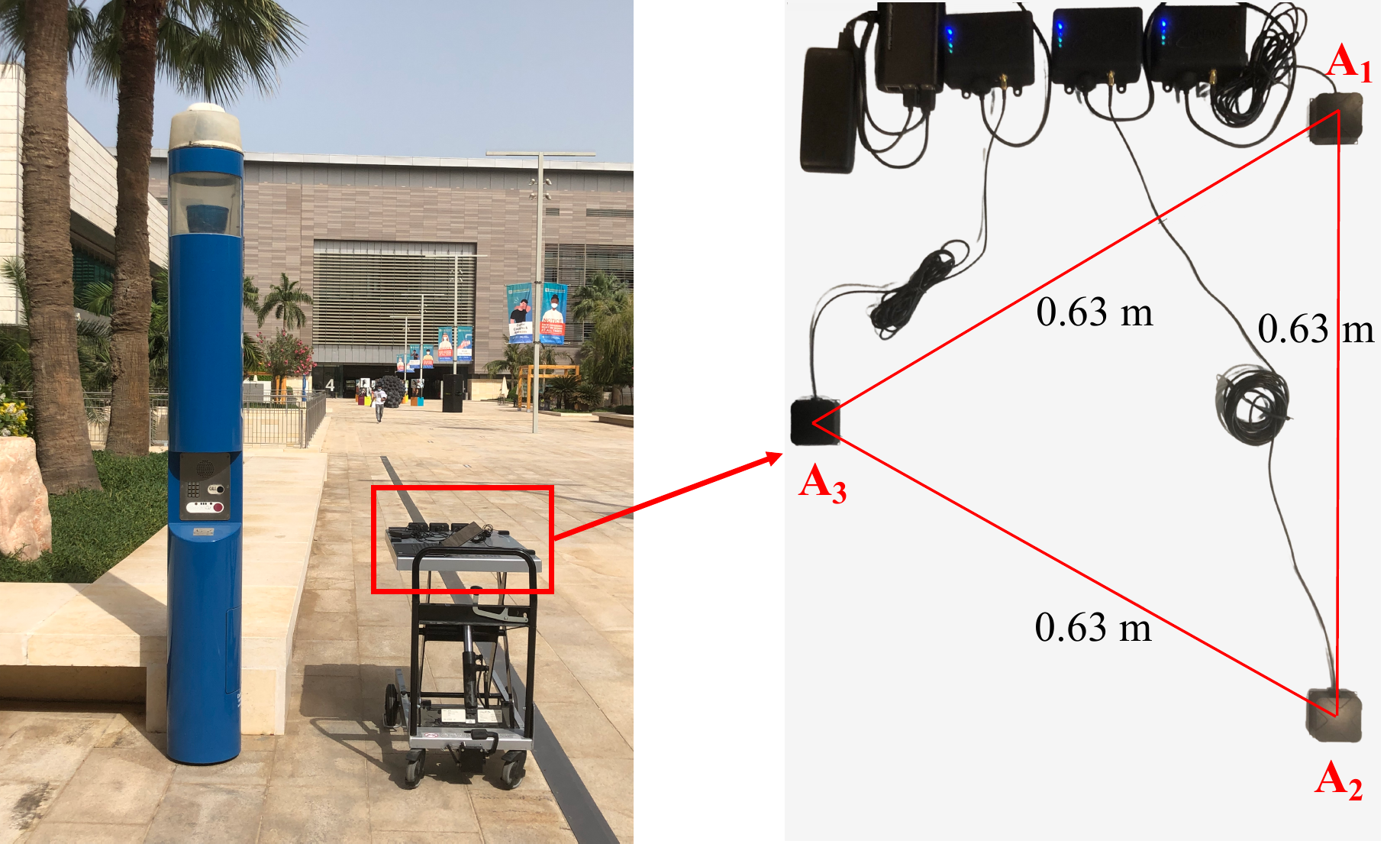}
\caption{GNSS antenna set-up.}
\label{fig:antenna_array}
\end{figure}

\begin{figure}[htbp]
\centering
\subfigure[1st experiment] { \label{fig:visiba} 
\includegraphics[width=0.46\columnwidth]{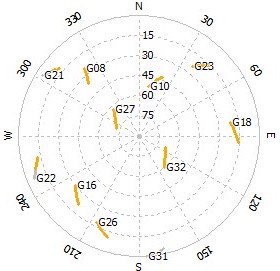}
} 
\subfigure[2nd experiment] { \label{fig:visibb} 
\includegraphics[width=0.46\columnwidth]{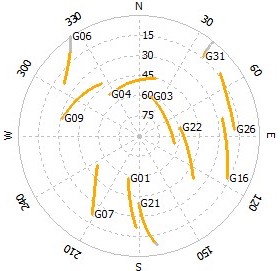} 
} 
\caption{Satellite visibility during the tests.}
\label{fig:visib}
\end{figure}

\begin{table}[htbp]
\centering
\caption{Success rate (\%) based on experimental data.}
\label{tab:expsuc}
\centering
\setlength{\tabcolsep}{1mm}{
\begin{tabular}{c| c| c}
\hline
&1st experiment &2nd experiment\\
   \hline
UC-AM/AC-AM &0.01 &50.06\\
\hline
RieMOCAD-LF &5.83 &86.81\\
\hline
MC-LAMBDA &6.75 &99.28\\
\hline
RieMOCAD-TF &\textbf{72.07} &\textbf{99.33}\\
\hline
\end{tabular}
}
\end{table} 
\begin{table}[htbp]
\centering
\caption{Number of integers in the search space.}
\label{tab:numexp}
\centering
\setlength{\tabcolsep}{1mm}{
\begin{tabular}{c| c| c}
\hline
&1st experiment &2nd experiment\\
   \hline
MC-LAMBDA &9891.17 &1118.65\\
\hline
RieMOCAD-TF &\textbf{8649.51} &\textbf{795.02}\\
\hline
\end{tabular}
}
\end{table}

\begin{figure*}[htbp]
\centering
\includegraphics[width=1.\textwidth]{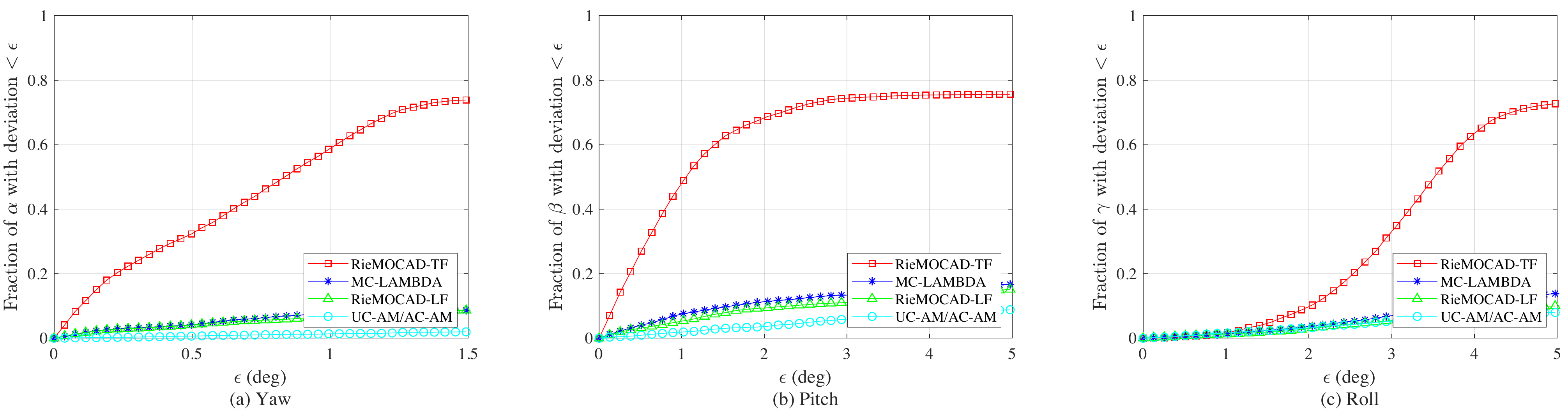}
\caption{Euler angle error distribution for the first experiment.}
\label{fig:firstexperiment}
\end{figure*}
\begin{figure*}[htbp]
\centering
\includegraphics[width=1.\textwidth]{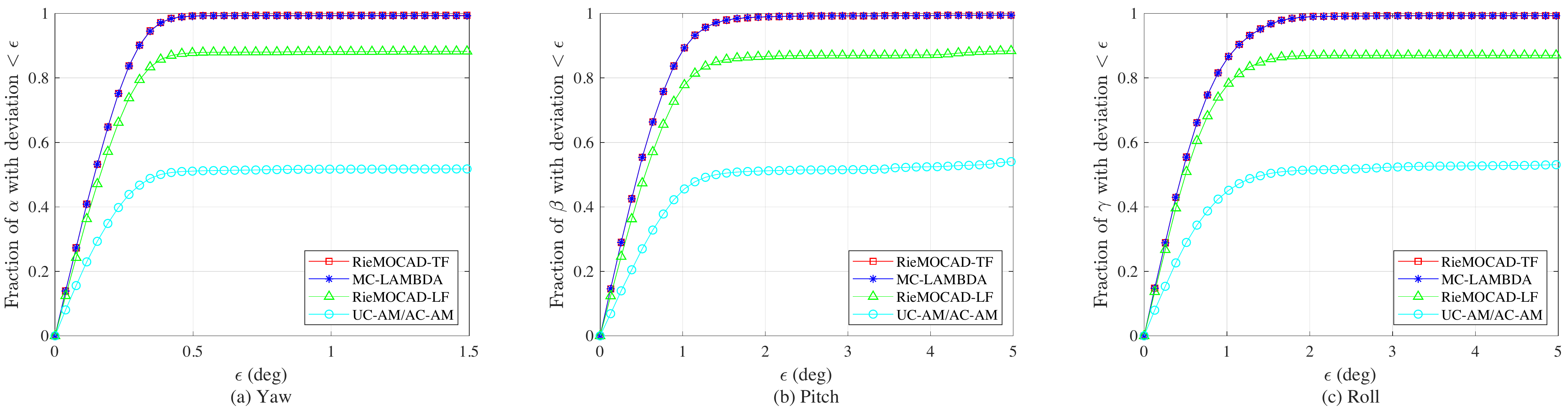}
\caption{Euler angle error distribution for the second experiment.}
\label{fig:secondexperiment}
\end{figure*}

\subsection{Experimental Evaluation}
In addition to the simulations, we evaluate the proposed method's performance using experimental data and compare it with other benchmarks. To estimate the ground truth precisely, we carried out two static experiments using three ANAVS multi-sensor modules. The receivers were arranged in an equilateral triangle with a 0.63-meter side, as shown in \fref{fig:antenna_array}. The receivers were firmly mounted on the platform, which remained stationary at positions (N$22^\circ 18^\prime 34''$, E$39^\circ 6^\prime 16''$) and (N$22^\circ 18^\prime 32''$, E$39^\circ 6^\prime 27''$) during the two experiments, respectively. These receivers can collect GNSS data at a 5 Hz sampling rate. The first experiment was performed from 18:35 to 19:03 (UTC) on 6 April 2021, with 8500 epochs of GNSS data collected. The other experiment was carried out from 20:34 to 21:50 (UTC) on 25 May 2021 to acquire 25500 epochs of GNSS measurements. The satellite visibility during the tests is shown in \fref{fig:visib}, with a 10-degree cut-off angle. We estimate the ground truth by using the data of the entire test period and fusing data from multiple sensors, including GNSS observations (GPS+GLONASS) and data from an inertial measurement unit (IMU).

\tref{tab:expsuc} shows the success rate of ambiguity resolution for both experiments. The estimations in the second experiment are much better since we fixed the platform at a place with good satellite visibility. The location of the first experiment was surrounded by obstacles, such as trees and buildings. We can also see that the AC-AM (the same as UC-AM) shows a poor success rate. As expected, RieMOCAD-LF offers a higher success rate compared to the AC-AM. The two tight-form implementations of the OC-AM, i.e., the MC-LAMBDA method and RieMOCAD-TF, beat other techniques in ambiguity resolution success rate, with RieMOCAD-TF performing the best. As aforesaid, the constrained search space in the integer domain decides the difference in ambiguity resolution performance between RieMOCAD-TF and the MC-LAMBDA method. The size of the search space, i.e., the number of integer candidates as summarized in \tref{tab:numexp}, reveals the low complexity of RieMOCAD-TF.

\fref{fig:firstexperiment} and \fref{fig:secondexperiment} show the cumulative errors of the attitude angles (Euler angles). Consistent with the success-rate performance described above, RieMOCAD-LF outperforms the AC-AM. On the other hand, RieMOCAD-TF outperforms the corresponding benchmark, the MC-LAMBDA method. The success rate, size of the search space, and the angle error distribution prove the reliability and feasibility of the proposed method.


\section{CONCLUSION}
We proposed an ambiguity resolution method for GNSS attitude determination. By modeling the GNSS attitude problem as an optimization on the Stiefel manifold, we alleviate the difficulties to resolve the unknown carrier-phase ambiguities. First, to replace the widely used least-squares solution, we calculate a high-quality float solution using Riemannian manifold optimization as an improved starting point for the integer search process. To take advantage of the improved intermediate solution, we propose decomposing the objective function at an arbitrary point instead of the widely used orthogonal decomposition. By leveraging Riemannian optimization, the improved float solution, and the proposed decomposition, we are able to search for the carrier-phase ambiguities in the integer domain with high efficiency and reliability. The presented algorithm has two variants: a loose form and a tight form known as RieMOCAD-LF and RieMOCAD-TF, respectively. They are compared to relevant state-of-the-art methods using both simulation and experimental data. We evaluate performance in terms of key performance indicators such as success rate, search space size, and attitude error distribution. The results confirm the feasibility and effectiveness of the proposed approach.

\section*{APPENDIX A}
\label{sec:apA}
\setcounter{equation}{0}
\setcounter{subsection}{0}
\renewcommand{\theequation}{A.\arabic{equation}}
\renewcommand{\thesubsection}{A.\arabic{subsection}}

\emph{Proof of \lref{le:decomfloat}:} We first vectorize all the unknown parameters in a single vector ${\bm{\theta}} = 
 \mathrm{vec} \left ( \left[\mathbf{R} \text{ } \mathbf{N} \right] \right )$ and define the objective function as
\begin{equation}
\mathcal{F}\! \left ( {\bm{\theta}}  \right ) = \left\| {\operatorname{vec}\!\left( {{\mathbf{Y}} - {\mathbf{AR}}{{\mathbf{X}}_b} - {\mathbf{BN}}} \right)} \right\|_{\mathbf{Q}_{\mathbf{Y}}^{-1}}^2.
\end{equation}
To achieve the LS solution ${\bm{\hat \theta}} = 
 \mathrm{vec} \left ( \left[\mathbf{\hat R}_{_\text{AC}} \text{ } \mathbf{\hat N}_{_\text{AC}} \right] \right )$, we let the first order derivative equal to zero \cite{6164274}
\begin{equation}
\mathcal{F}'\! \left ( {\bm{\hat \theta}}  \right ) = {\mathbf{0}},
\end{equation}
that is
\begin{equation}
\mathbf{M} \cdot {\bm{\hat \theta}} = \begin{bmatrix}
  \mathbf{X}_{b}\mathbf{P}^{-1}\otimes \mathbf{A}^{\mathrm{T} }\mathbf{Q}_\mathbf{y}^{-1}\\
  \mathbf{P}^{-1}\otimes \mathbf{B}^{\mathrm{T} }\mathbf{Q}_\mathbf{y}^{-1}
\end{bmatrix}\mathrm{vec} \left ( \mathbf{Y}  \right ),
\label{eq:A3}
\end{equation}
with
\begin{equation}
\resizebox{.86\hsize}{!}{
$
\mathbf{M} = \begin{bmatrix}
  \mathbf{X}_{b}\mathbf{P}^{-1}\otimes \mathbf{A}^{\mathrm{T} }\mathbf{Q}_\mathbf{y}^{-1}\mathbf{B}  & \mathbf{X}_{b}\mathbf{P}^{-1}\mathbf{X}_{b}^{\mathrm{T} }\otimes \mathbf{A}^{\mathrm{T} }\mathbf{Q}_\mathbf{y}^{-1}\mathbf{A}\\
  \mathbf{P}^{-1}\otimes \mathbf{B}^{\mathrm{T} }\mathbf{Q}_\mathbf{y}^{-1}\mathbf{B}  & \mathbf{P}^{-1}\mathbf{X}_{b}^{\mathrm{T} }\otimes \mathbf{B}^{\mathrm{T} }\mathbf{Q}_\mathbf{y}^{-1}\mathbf{A}
\end{bmatrix},
$}
\label{eq:A4}
\end{equation}
\begin{equation}
{\mathbf{Q}_{\mathbf{Y}}} = \mathbf{P} \otimes {\mathbf{Q}_{\mathbf{y}}} = \begin{bmatrix}
  1 &0.5  &\cdots  &0.5 \\
  0.5&1  &\cdots  &0.5 \\
  &  &\ddots  & \\
  0.5&0.5  &\cdots  &1
\end{bmatrix} \otimes {\mathbf{Q}_{\mathbf{y}}},
\end{equation}
where ${\mathbf{Q}_{\mathbf{y}}}$ is the covariance matrix of GNSS observations for a single baseline, and $\otimes$ indicates the Kronecker product.

Since $\mathcal{F}\! \left ( {\bm{\theta}}  \right )$ is a quadratic form of $ {\bm{\theta}}$, its third or higher order derivatives are all equal to zero, i.e., $\mathcal{F}^{\left(3 \right)}\! \left ( {\bm{\theta}}  \right ) = \mathcal{F}^{\left(4 \right)}\! \left ( {\bm{\theta}}  \right ) = \cdots = \mathcal{F}^{\left(n \right)}\! \left ( {\bm{\theta}}  \right ) = \mathbf{0}$. Then, we can expand $\mathcal{F}\! \left ( {\bm{\theta}}  \right )$ as a Taylor series at ${\bm{\hat \theta}}$
\begin{equation}
\mathcal{F}\! \left ( {\bm{\theta}}  \right ) = \mathcal{F}\! \left ( {\bm{\hat \theta}}  \right ) + \frac{1}{2}\left ( {\bm{\theta}} - {\bm{\hat \theta}}  \right )^{\text{T}}\mathcal{F}''\! \left ( {\bm{\hat \theta}}  \right )\left ( {\bm{\theta}} - {\bm{\hat \theta}}  \right ),
\label{aeq:objf}
\end{equation}
with
\begin{equation*}
\mathcal{F}''\! \left ( {\bm{\hat \theta}}  \right ) = 2{\mathbf{M}}.
\end{equation*}
To rewrite the second term on the right-hand side of \eref{aeq:objf} as a sum of squares terms, we define the block-triangular transformation
\begin{equation}
\mathbf{T}_1 = \begin{bmatrix}
  \mathbf{I} &\mathbf{O}\\
  \mathbf{I} \otimes \mathbf{B}^{+}\mathbf{A} &\mathbf{I}
\end{bmatrix}.
\end{equation}
Then we obtain
\begin{equation}
\mathbf{T}_1\left ( {\bm{\theta}} - {\bm{\hat \theta}}  \right ) = \begin{bmatrix}
 \mathrm{vec} \!\!\left ( \!\mathbf{R} -\mathbf{\hat R}_{_\text{AC}}\!\right ) \\
\mathrm{vec} \!\!\left (\! \mathbf{N} - \mathbf{\hat N}_{_\text{AC}}\!\left (\mathbf{R} \right ) \!\right )
\end{bmatrix},
\end{equation}
and
\begin{equation}
\mathbf{T}_1^{-\text{T}}\mathbf{M}\mathbf{T}_1^{-1} = \begin{bmatrix}
 {\mathbf{Q}_{{\mathbf{\hat R}_{_\text{AC}}}\!{\mathbf{\hat R}_{_\text{AC}}}}^{-1}} &\mathbf{O}\\
\mathbf{O} & {\mathbf{Q}_{{\mathbf{\hat N}_{_\text{AC}}}\!\left({\mathbf{R}}\right){\mathbf{\hat N}_{_\text{AC}}}\!\left({\mathbf{R}}\right)}^{-1}}
\end{bmatrix}.
\end{equation}
Based on the transformation, we have
\begin{equation}
\resizebox{.83\hsize}{!}{
$
\begin{aligned}
\mathcal{F}\! \left ( {\bm{\theta}}  \right )
 = &\mathcal{F}\! \left ( {\bm{\hat \theta}}  \right ) \!+\! \left ( \!{\bm{\theta}} \!-\! {\bm{\hat \theta}}  \!\right )^{\text{T}}\mathbf{M}\left ( \!{\bm{\theta}} \!- \!{\bm{\hat \theta}} \! \right )\\
 = &\mathcal{F}\! \left ( {\bm{\hat \theta}}  \right ) \!+\! \left ( \!{\bm{\theta}} \!-\! {\bm{\hat \theta}}  \!\right )^{\text{T}}\mathbf{T}_1^{\text{T}}\mathbf{T}_1^{-\text{T}}\mathbf{M}\mathbf{T}_1^{-1} \mathbf{T}_1\left ( \!{\bm{\theta}} \!- \!{\bm{\hat \theta}}  \!\right )\\
  = &\mathcal{F}\! \left ( {\bm{\hat \theta}}  \right ) \!+\! \left\| \mathrm{vec} \!\left ( \mathbf{R} \!-\!\mathbf{\hat R}_{_\text{AC}}\!\right )\! \right\|_{\mathbf{Q}_{{\mathbf{\hat R}_{_\text{AC}}}\!{\mathbf{\hat R}_{_\text{AC}}}}^{-1}}^2 \\ &+\! \left\| \mathrm{vec} \!\left ( \mathbf{N} \!-\!\mathbf{\hat N}_{_\text{AC}}\!\!\left (\mathbf{R} \right )\!\right )\! \right\|_{\mathbf{Q}_{{\mathbf{\hat N}_{_\text{AC}}}\!\left({\mathbf{R}}\right){\mathbf{\hat N}_{_\text{AC}}}\!\left({\mathbf{R}}\right)}^{-1}}^2.
\end{aligned}
$}
\end{equation}

\section*{APPENDIX B}
\label{sec:apB}
\setcounter{equation}{0}
\setcounter{subsection}{0}
\renewcommand{\theequation}{B.\arabic{equation}}
\renewcommand{\thesubsection}{B.\arabic{subsection}}

\emph{Proof of \lref{th:decom}:} Consider an arbitrary point ${\bm{\bar \theta}} =
 \mathrm{vec} \left ( \left[\mathbf{\bar R} \text{ } \mathbf{\bar N} \right] \right )$. According to \eref{aeq:objf}, we have
\begin{equation}
\resizebox{.83\hsize}{!}{
$
\begin{aligned}
\mathcal{F}\! \left ( {\bm{\theta}}  \right ) = &\mathcal{F}\! \left ( {\bm{\hat \theta}}  \right ) \!\!+\!\! \left ( \!{\bm{\theta}} \!-\! {\bm{\bar \theta}} \!+\! {\bm{\bar \theta}} \!-\! {\bm{\hat \theta}} \! \right )^{\text{T}}\!\mathbf{M}\!\left (\! {\bm{\theta}} \!- \!{\bm{\bar \theta}} \!+\! {\bm{\bar \theta}} \!- \!{\bm{\hat \theta}}  \!\right )\\
 = &\mathcal{F}\! \left ( {\bm{\hat \theta}}  \right ) \!+\! \left (\! {\bm{\bar \theta}}\! -\! {\bm{\hat \theta}} \! \right )^{\text{T}}\!\mathbf{M}\!\left (\! {\bm{\bar \theta}}\! -\! {\bm{\hat \theta}}  \!\right ) \\
 &+ \!\!\left ( \!{\bm{\theta}} \!- \!{\bm{\bar \theta}}   \!\right )^{\text{T}}\!\mathbf{M}\!\left ( \!{\bm{\theta}} \!- \!{\bm{\bar \theta}}  \!\right )  \!\!+\!\! 2\left ( \!{\bm{\theta}} \!-\! {\bm{\bar \theta}} \!\right )^{\text{T}}\!\mathbf{M}\!\left ( \!{\bm{\bar \theta}} \!- \!{\bm{\hat \theta}}  \!\right )\\
 = &\mathcal{F}\! \left ( {\bm{\bar \theta}}  \right ) \!\! +\!\! \left (\! {\bm{\theta}} \!-\! {\bm{\bar \theta}}  \! \right )^{\text{T}}\!\mathbf{M}\!\left ( \!{\bm{\theta}}\!-\! {\bm{\bar \theta}} \! \right ) \!\!+ \!\!2\left ( \!{\bm{\theta}} \!-\! {\bm{\bar \theta}} \!\right )^{\text{T}}\!\mathbf{M}\!\left (\! {\bm{\bar \theta}} \!-\! {\bm{\hat \theta}}  \!\right ).
\end{aligned}
$}
\label{aeq:thpf1}
\end{equation}
Again, we define the block-triangular transformation
\begin{equation}
\mathbf{T}_2 = \begin{bmatrix}
  \mathbf{I} &\left(\mathbf{X}_b^{+} \right)^{\text{T}} \otimes \mathbf{A}^{+}\mathbf{B}\\
  \mathbf{O} &\mathbf{I}
\end{bmatrix}.
\end{equation}
Then we have
\begin{equation}
\mathbf{T}_2\left ( {\bm{\theta}} - {\bm{\bar \theta}}  \right ) = \begin{bmatrix}
 \mathrm{vec} \!\!\left ( \!\mathbf{R} -\mathbf{\bar R}\!\left (\mathbf{N} \right )\!\right ) \\
\mathrm{vec} \!\!\left (\! \mathbf{N} - \mathbf{\bar N} \!\right )
\end{bmatrix},
\end{equation}
\begin{equation}
\mathbf{T}_2\left ( {\bm{\bar \theta}} - {\bm{\hat \theta}}  \right ) = \begin{bmatrix}
 \mathrm{vec} \!\!\left ( \!\mathbf{\bar R} -\mathbf{\hat R}_{_\text{AC}}\!\left (\mathbf{\bar N} \right )\!\right ) \\
\mathrm{vec} \!\!\left (\! \mathbf{\bar N} - \mathbf{\hat N}_{_\text{AC}} \!\right )
\end{bmatrix},
\end{equation}
\begin{equation}
\mathbf{T}_2^{-\text{T}}\mathbf{M}\mathbf{T}_2^{-1} = \begin{bmatrix}
 {\mathbf{Q}_{{\mathbf{\hat R}_{_\text{AC}}\!\left({\mathbf{N}}\right)}\!{\mathbf{\hat R}_{_\text{AC}}\!\left({\mathbf{N}}\right)}}^{-1}} &\mathbf{O}\\
\mathbf{O} & {\mathbf{Q}_{{\mathbf{\hat N}_{_\text{AC}}}{\mathbf{\hat N}_{_\text{AC}}}}^{-1}}
\end{bmatrix}.
\end{equation}
Therefore, we obtain
\begin{equation}
\resizebox{.83\hsize}{!}{
$
\begin{aligned}
&\left ( {\bm{\theta}} - {\bm{\bar \theta}}   \right )^{\text{T}}\mathbf{M}\left ( {\bm{\theta}} \!-\! {\bm{\bar \theta}}  \right ) \\
= &\left ( {\bm{\theta}} - {\bm{\bar \theta}}   \right )^{\text{T}}\mathbf{T}_2^{\text{T}}\mathbf{T}_2^{-\text{T}}\mathbf{M}\mathbf{T}_2^{-1}\mathbf{T}_2\left ( {\bm{\theta}} \!-\! {\bm{\bar \theta}}  \right ) \\
= &\! \left\| \!\mathrm{vec}\! \left ( \!\mathbf{N} \!-\!\mathbf{\bar N} \!\right ) \!\right\|^2_ {\mathbf{Q}_{{\mathbf{\hat N}_{_\text{AC}}}{\mathbf{\hat N}_{_\text{AC}}}}^{-1}} \!\!\!+\! \left\| \!\mathrm{vec} \!\left ( \!\mathbf{R} \!-\!\mathbf{\bar R}\left (\!\mathbf{N}\! \right )\!\right ) \!\right\|^2_{\mathbf{Q}_{{\mathbf{\hat R}_{_\text{AC}}\!\left({\mathbf{N}}\right)}\!{\mathbf{\hat R}_{_\text{AC}}\!\left({\mathbf{N}}\right)}}^{-1}},
\end{aligned}
$}
\label{aeq:thpf2}
\end{equation}
and
\begin{equation}
\resizebox{.8\hsize}{!}{
$
\begin{aligned}
&\left ( {\bm{\theta}} \!- \!{\bm{\bar \theta}} \right )^{\text{T}}\mathbf{M}\left ( {\bm{\bar \theta}} - {\bm{\hat \theta}}  \right ) \\
= &\left ( {\bm{\theta}} \!- \!{\bm{\bar \theta}} \right )^{\text{T}}\mathbf{T}_2^{\text{T}}\mathbf{T}_2^{-\text{T}}\mathbf{M}\mathbf{T}_2^{-1}\mathbf{T}_2\left ( {\bm{\bar \theta}} - {\bm{\hat \theta}}  \right ) \\
=  &\mathrm{vec} \!\left ( \!\mathbf{N} \!-\!\mathbf{\bar N} \!\right )^{\text{T}}{\mathbf{Q}_{{\mathbf{\hat N}_{_\text{AC}}}{\mathbf{\hat N}_{_\text{AC}}}}^{-1}} \! \mathrm{vec}\! \left ( \!\mathbf{\bar N} \!-\!\mathbf{\hat N}_{_\text{AC}} \!\right ) \\
&+ \mathrm{vec}\! \left (\! \mathbf{R} \!-\!\mathbf{\bar R}\!\left (\!\mathbf{N}\! \right )\!\right )^{\text{T}}{\mathbf{Q}_{{\mathbf{\hat R}_{_\text{AC}}\!\left({\mathbf{N}}\right)}\!{\mathbf{\hat R}_{_\text{AC}}\!\left({\mathbf{N}}\right)}}^{-1}}\mathrm{vec} \!\left (\! \mathbf{\bar R} \!-\!\mathbf{\hat R}_{_\text{AC}}\!\left (\!\mathbf{\bar N} \!\right )\!\right ).
\end{aligned}
$}
\label{aeq:thpf3}
\end{equation}
The combination of \eref{aeq:thpf1}, \eref{aeq:thpf2} and \eref{aeq:thpf3} concludes the proof.

\section*{APPENDIX C}
\label{sec:apC}
\setcounter{equation}{0}
\setcounter{subsection}{0}
\renewcommand{\theequation}{C.\arabic{equation}}
\renewcommand{\thesubsection}{C.\arabic{subsection}}

\emph{Proof of \lref{le:decom}:} According to \lref{th:decom}, $\mathcal{F}\! \left ( {\bm{\theta}}  \right )$ can be expressed as 
\begin{equation}
\resizebox{.84\hsize}{!}{
$
\begin{aligned}
\mathcal{F}\! \left ( {\bm{\theta}}  \right )
 = &\mathcal{F}\! \left ( {\bm{\bar \theta}}  \right ) \!\! +\!\! \left\| \!\mathrm{vec}\! \left ( \!\mathbf{N} \!-\!\mathbf{\bar N} \!\right ) \!\right\|^2_ {\mathbf{Q}_{{\mathbf{\hat N}_{_\text{AC}}}{\mathbf{\hat N}_{_\text{AC}}}}^{-1}} \!\!\!+\! \left\| \!\mathrm{vec} \!\left ( \!\mathbf{R} \!-\!\mathbf{\bar R}\!\left (\!\mathbf{N}\! \right )\!\right ) \!\right\|^2_{\mathbf{Q}_{{\mathbf{\hat R}_{_\text{AC}}\!\left({\mathbf{N}}\right)}\!{\mathbf{\hat R}_{_\text{AC}}\!\left({\mathbf{N}}\right)}}^{-1}}\\
 &+ \!2\mathrm{vec} \!\left ( \!\mathbf{N} \!-\!\mathbf{\bar N} \!\right )^{\text{T}}{\mathbf{Q}_{{\mathbf{\hat N}_{_\text{AC}}}{\mathbf{\hat N}_{_\text{AC}}}}^{-1}} \! \mathrm{vec}\! \left ( \!\mathbf{\bar N} \!-\!\mathbf{\hat N}_{_\text{AC}} \!\right ) \\
&+\! 2\mathrm{vec}\! \left (\! \mathbf{R} \!-\!\mathbf{\bar R}\!\left (\!\mathbf{N}\! \right )\!\right )^{\text{T}}{\mathbf{Q}_{{\mathbf{\hat R}_{_\text{AC}}\!\left({\mathbf{N}}\right)}\!{\mathbf{\hat R}_{_\text{AC}}\!\left({\mathbf{N}}\right)}}^{-1}}\mathrm{vec} \!\left (\! \mathbf{\bar R} \!-\!\mathbf{\hat R}_{_\text{AC}}\!\left (\!\mathbf{\bar N} \!\right )\!\right )\\
&+ \!\left\| \!\mathrm{vec} \!\left (\! \mathbf{\bar R} \!-\!\mathbf{\hat R}\left (\!\mathbf{\bar N} \!\right )\!\right )\! \right\|^2_{\mathbf{Q}_{{\mathbf{\hat R}_{_\text{AC}}\!\left({\mathbf{N}}\right)}\!{\mathbf{\hat R}_{_\text{AC}}\!\left({\mathbf{N}}\right)}}^{-1}}\! -\! \left\| \!\mathrm{vec} \!\left (\! \mathbf{\bar R} \!-\!\mathbf{\hat R}\!\left (\!\mathbf{\bar N} \!\right )\!\right ) \!\right\|^2_{\mathbf{Q}_{{\mathbf{\hat R}_{_\text{AC}}\!\left({\mathbf{N}}\right)}\!{\mathbf{\hat R}_{_\text{AC}}\!\left({\mathbf{N}}\right)}}^{-1}}.
\end{aligned}
$}
\label{aeq:lempf1}
\end{equation}
Note that 
\begin{equation}
\resizebox{.84\hsize}{!}{
$
\begin{aligned}
&\left\| \!\mathrm{vec} \!\left ( \!\mathbf{R} \!-\!\mathbf{\hat R}_{_\text{AC}}\!\left (\!\mathbf{N}\! \right )\!\right ) \!\right\|^2_{\mathbf{Q}_{{\mathbf{\hat R}_{_\text{AC}}\!\left({\mathbf{N}}\right)}\!{\mathbf{\hat R}_{_\text{AC}}\!\left({\mathbf{N}}\right)}}^{-1}}\\
=&\left\| \!\mathrm{vec} \!\left ( \!\mathbf{R} \!-\!\mathbf{\bar R}\left (\!\mathbf{N}\! \right ) \!+\!\mathbf{\bar R}\left (\!\mathbf{N}\! \right )\!-\!\mathbf{\hat R}_{_\text{AC}}\!\left (\!\mathbf{N}\! \right )\!\right ) \!\right\|^2_{\mathbf{Q}_{{\mathbf{\hat R}_{_\text{AC}}\!\left({\mathbf{N}}\right)}\!{\mathbf{\hat R}_{_\text{AC}}\!\left({\mathbf{N}}\right)}}^{-1}}\\
=&\left\| \!\mathrm{vec} \!\left ( \!\mathbf{R} \!-\!\mathbf{\bar R}\left (\!\mathbf{N}\! \right )\!\right ) \!\right\|^2_{\mathbf{Q}_{{\mathbf{\hat R}_{_\text{AC}}\!\left({\mathbf{N}}\right)}\!{\mathbf{\hat R}_{_\text{AC}}\!\left({\mathbf{N}}\right)}}^{-1}} + \!\left\| \!\mathrm{vec} \!\left (\! \mathbf{\bar R}\!\!\left (\!\mathbf{N} \!\right ) \!-\!\mathbf{\hat R}\left (\!\mathbf{N} \!\right )\!\right )\! \right\|^2_{\mathbf{Q}_{{\mathbf{\hat R}_{_\text{AC}}\!\left({\mathbf{N}}\right)}\!{\mathbf{\hat R}_{_\text{AC}}\!\left({\mathbf{N}}\right)}}^{-1}}\\
&+\! 2\mathrm{vec}\! \left (\! \mathbf{R} \!-\!\mathbf{\bar R}\!\left (\!\mathbf{N}\! \right )\!\right )^{\text{T}}{\mathbf{Q}_{{\mathbf{\hat R}_{_\text{AC}}\!\left({\mathbf{N}}\right)}\!{\mathbf{\hat R}_{_\text{AC}}\!\left({\mathbf{N}}\right)}}^{-1}}\mathrm{vec} \!\left (\! \mathbf{\bar R}\!\left (\!\mathbf{N} \!\right ) \!-\!\mathbf{\hat R}_{_\text{AC}}\!\left (\!\mathbf{N} \!\right )\!\right ).
\end{aligned}
$}
\end{equation}
Since 
\begin{equation}
\mathbf{\bar R} \!-\!\mathbf{\hat R}_{_\text{AC}}\!\left (\!\mathbf{\bar N} \!\right ) = \mathbf{\bar R}\!\left (\!\mathbf{N} \!\right ) \!-\!\mathbf{\hat R}_{_\text{AC}}\!\left (\!\mathbf{N} \!\right ),
\end{equation}
we have
\begin{equation}
\resizebox{.84\hsize}{!}{
$
\begin{aligned}
&\left\| \!\mathrm{vec} \!\left ( \!\mathbf{R} \!-\!\mathbf{\hat R}_{_\text{AC}}\!\left (\!\mathbf{N}\! \right )\!\right ) \!\right\|^2_{\mathbf{Q}_{{\mathbf{\hat R}_{_\text{AC}}\!\left({\mathbf{N}}\right)}\!{\mathbf{\hat R}_{_\text{AC}}\!\left({\mathbf{N}}\right)}}^{-1}}\\
=&\left\| \!\mathrm{vec} \!\left ( \!\mathbf{R} \!-\!\mathbf{\bar R}\left (\!\mathbf{N}\! \right )\!\right ) \!\right\|^2_{\mathbf{Q}_{{\mathbf{\hat R}_{_\text{AC}}\!\left({\mathbf{N}}\right)}\!{\mathbf{\hat R}_{_\text{AC}}\!\left({\mathbf{N}}\right)}}^{-1}} + \!\left\| \!\mathrm{vec} \!\left (\! \mathbf{\bar R} \!-\!\mathbf{\hat R}\left (\!\mathbf{\bar N} \!\right )\!\right )\! \right\|^2_{\mathbf{Q}_{{\mathbf{\hat R}_{_\text{AC}}\!\left({\mathbf{N}}\right)}\!{\mathbf{\hat R}_{_\text{AC}}\!\left({\mathbf{N}}\right)}}^{-1}}\\
&+\! 2\mathrm{vec}\! \left (\! \mathbf{R} \!-\!\mathbf{\bar R}\!\left (\!\mathbf{N}\! \right )\!\right )^{\text{T}}{\mathbf{Q}_{{\mathbf{\hat R}_{_\text{AC}}\!\left({\mathbf{N}}\right)}\!{\mathbf{\hat R}_{_\text{AC}}\!\left({\mathbf{N}}\right)}}^{-1}}\mathrm{vec} \!\left (\! \mathbf{\bar R} \!-\!\mathbf{\hat R}_{_\text{AC}}\!\left (\!\mathbf{\bar N} \!\right )\!\right ).
\end{aligned}
$}
\label{aeq:lempf2}
\end{equation}
Finally, \eref{aeq:lempf2} combined with \eref{aeq:lempf1} concludes the proof.




\bibliographystyle{IEEEtran}
\bibliography{IEEEabrv,IEEEfull}

\end{document}